%% file: ArxivV3-RulingOutHigherOrderInterference.tex
\documentclass[11pt,british]{article}
\usepackage[T1]{fontenc}
\usepackage[latin9]{inputenc}
\usepackage{amsmath}
\usepackage{amsthm}
\usepackage{amssymb}
\usepackage{color}
\usepackage[unicode=true,pdfusetitle,
 bookmarks=true,bookmarksnumbered=false,bookmarksopen=false,
 breaklinks=true,pdfborder={0 0 0},pdfborderstyle={},backref=false,colorlinks=true]
 {hyperref}

\makeatletter
\numberwithin{equation}{section}
\theoremstyle{plain}
\newtheorem{thm}{\protect\theoremname}[section]
  \theoremstyle{definition}
  \newtheorem{Definition}[thm]{\protect\definitionname}
  \theoremstyle{plain}
  \newtheorem{ax}[thm]{\protect\axiomname}
  \theoremstyle{plain}
  \newtheorem{Proposition}[thm]{\protect\propositionname}
  \theoremstyle{plain}
  \newtheorem{Lemma}[thm]{\protect\lemmaname}
  \theoremstyle{plain}

\@ifundefined{date}{}{\date{}}

\makeatother
\input{myQcircuit}
\makeatletter

\newcommand{\rA}{\mathrm{A}}
\newcommand{\rB}{\mathrm{B}}
\newcommand{\rC}{\mathrm{C}}
\newcommand{\rD}{\mathrm{D}}
\newcommand{\rS}{\mathrm{S}}
\newcommand{\cA}{\mathcal{A}}
\newcommand{\cB}{\mathcal{B}}
\newcommand{\cC}{\mathcal{C}}
\newcommand{\cI}{\mathcal{I}}
\newcommand{\cU}{\mathcal{U}}

\makeatother
\usepackage[margin=2.8cm, a4paper]{geometry}
\usepackage{babel}
  \providecommand{\axiomname}{Axiom}
  \providecommand{\corollaryname}{Corollary}
  \providecommand{\definitionname}{Definition}
  \providecommand{\lemmaname}{Lemma}
  \providecommand{\propositionname}{Proposition}
\providecommand{\theoremname}{Theorem}

\begin{document}


\begingroup
\centering
{\Large\textbf{Ruling out higher-order interference from purity principles} \\[1.5em]
 \normalsize Howard Barnum\textsuperscript{$\dagger,\mathsection$} Ciar\'{a}n M. Lee\textsuperscript{$\ddagger$,}\footnote{Electronic address: ciaran.lee@ucl.ac.uk} Carlo Maria Scandolo\textsuperscript{$\ast,$}\footnote{Electronic address: carlomaria.scandolo@cs.ox.ac.uk} and John H. Selby\textsuperscript{$\ast,\bullet$,}\footnote{Electronic address: john.selby08@imperial.ac.uk}
}
\\[1em]
\small
\it \textsuperscript{$\dagger$}   QMATH, Department of Mathematical Sciences, University of Copenhagen, Denmark.\\ \textsuperscript{$\mathsection$} Department of Physics and Astronomy, University of New Mexico, Albuquerque, USA.\\
  \it \textsuperscript{$\ddagger$} Department of Physics and Astronomy, University College London, UK. \\ \it \textsuperscript{$\ast$}  University of Oxford, Department of Computer Science, OX1 3QD, UK. \\ \it \textsuperscript{$\bullet$} Imperial College London,  London SW7 2AZ, UK. \\

\endgroup

\begin{abstract}
As first noted by Rafael Sorkin, there is a limit to quantum interference. The interference pattern formed in a multi-slit experiment is a function of the interference patterns formed between pairs of slits; there are no genuinely
new features resulting from considering three slits instead of two. Sorkin has introduced a hierarchy of mathematically conceivable \emph{higher-order} interference behaviours, where classical theory lies at the first level of this hierarchy and quantum theory theory at the second. Informally, the order in this hierarchy corresponds to the number of slits on which the interference pattern has an irreducible dependence. Many authors have wondered why quantum interference is limited to the second level of this hierarchy. Does the existence of higher-order interference violate some natural physical principle that we believe should be fundamental? In the current work we show that  such principles can be found which limit interference behaviour to second-order, or ``quantum-like'', interference, but that do not restrict us to the entire quantum formalism. We work within the operational framework of generalised probabilistic theories, and prove that any theory satisfying Causality, Purity Preservation, Pure Sharpness, and Purification---four principles that formalise the fundamental character of purity in nature---exhibits at most second-order interference. Hence these theories are, at least conceptually, very ``close'' to quantum theory. Along the way we show that systems in such theories correspond to Euclidean Jordan algebras. Hence, they are self-dual and, moreover, multi-slit experiments in such theories are described by pure projectors.
\end{abstract}
\section{Introduction}

Described by Feynman as ``impossible, \emph{absolutely} impossible, to explain in any classical way''
 \cite{Feynman} (volume 1, chapter 37), quantum interference is a distinctive signature of non-classicality. However, as first noted by Rafael Sorkin \cite{Sorkin1,Sorkin2}, there is a limit to this interference; in contrast to the case of two slits, the interference pattern formed in a three slit experiment \emph{can} be written as a linear combination of two and one slit patterns. Sorkin has introduced a hierarchy of mathematically conceivable \emph{higher-order} interference behaviours, where classical theory lies at the first level of this hierarchy and quantum theory theory at the second. Informally, the order in this hierarchy corresponds to the number of slits on which the interference pattern has an irreducible dependence.

Many authors have wondered why quantum interference is limited to the second level of this hierarchy \cite{Sorkin1,Lee-Selby-interference,Lee-Selby-Grover,Barnum-interference,Ududec-3slits,UBEinterferenceUnpublished, CozThesis, Niestegge,Control-reversible,Density-cubes,bolotin2016ongoing}. Does the existence of higher-order interference violate some natural physical principle that we believe should be fundamental \cite{lee2017no}? In the current work we show that {such natural} principles can be found which limit interference behaviour to second-order, or ``quantum-like'', interference, but that do not restrict us to the entire quantum formalism.

We work in the framework of general probabilistic theories \cite{Barrett,hardy2011,Chiribella-purification,Chiribella-informational,Hardy-informational-1,Barnum-1,Barnum-2,Brukner,Masanes-physical-derivation,chiribella2016quantum,lee2015computation,lee2016information,lee2016bounds,landscape}.
This framework is general enough to accommodate essentially arbitrary operational theories, where an operational theory specifies a set of laboratory devices which can be connected together in different ways, and assigns probabilities to different experimental outcomes. Investigating how the structural and information-theoretic features of a given theory in this framework depend on different physical principles deepens our physical and intuitive understanding of such features. Indeed, many authors \cite{Hardy-informational-1,Chiribella-informational,Hardy-informational-2,Brukner,Masanes-physical-derivation} have derived the entire structure of finite-dimensional quantum theory from simple information-theoretic axioms---reminiscent of Einstein's derivation of special relativity from two simple physical principles. So far, ruling out higher-order interference has required thermodynamic arguments. Indeed, by combining the results and axioms of Refs.~\cite{TowardsThermo,Barnum-thermo}, higher-order interference could be ruled out in theories satisfying the combined axioms. In this paper we show that we can prove this in a more direct way from first principles, using only the axioms of Ref.~\cite{TowardsThermo}.

Many experimental investigations have searched for divergences from quantum theory by looking for higher-order interference \cite{sinha2010ruling,sinha2015superposition,park2012three,kauten2015obtaining,jin2017experimental}. These experiments involved passing a particle through a physical barrier with multiple slits and comparing the interference patterns formed on a screen behind the barrier when different subsets of slits are closed. Given this set-up, one would expect that the physical theory being tested should possess transformations that correspond to the action of blocking certain subsets of slits. Moreover, blocking all but two subsets of slits should not affect states which can pass through either slit. This intuition suggests that these transformations should correspond to projectors.

Many operational probabilistic theories do not possess such a natural mathematical interpretation of multi-slit experiments; indeed many theories do not admit well{-}defined projectors \cite{Lee-Selby-interference}. Here, we show that there exist natural information-theoretic principles that both imply the existence of the projector structure{,} and rule out third-, and higher-, order interference. The principles that ensure this structure are Causality, Purity Preservation, Pure Sharpness, and Purification. These formalise intuitive ideas about the fundamental role of purity in nature. More formally, we show that such theories possess a self-dualising inner product{,} and that there exist pure projectors which represent the opening and closing of slits in a multi-slit experiment. Barnum, M\"uller and Ududec  have shown that in any self-dual theory in which such projectors exist for every face, if projectors map pure states to pure states, then there can be at most second-order interference \cite{Barnum-interference} (Proposition 29).  The conjunction of our new results and the principle of Purity Preservation implies the conditions of Barnum et al.'s proposition. Hence sharp theories with purification do not exhibit higher-order interference.  In fact we prove a stronger result, that the systems in such theories are Euclidean Jordan {a}lgebras which have been studied in quantum foundations \cite{Ududec-3slits,Barnum-interference,BarnumGraydonWilceCCEJA}.

This paper is organised as follows. In Section~\ref{framework} we review the basics of the operational probabilistic theory framework. In Section~\ref{higher-order interference} we formally define higher-order interference. In Section~\ref{sharp theories with purification} we define sharp theories with purification and review relevant known results. In Section~\ref{proofs} we present and prove our new results. Finally, in Section~\ref{end}, we offer some suggestions on how new experiments might be devised {to observe} higher-order interference.

\section{Framework} \label{framework}

{We will describe theories in the framework of operational-probabilistic
theories (OPTs) \cite{Chiribella-purification,Chiribella-informational,hardy2011,Hardy-informational-2,hardy2013,Chiribella14,QuantumFromPrinciples},}
arising from the marriage of category theory \cite{Abramsky2004,Coecke-Kindergarten,Coecke-Picturalism,Selinger,Coecke2016,Coeckebook}
with probabilities. The foundation of this framework is the idea that any successful physical theory must provide an account of experimental data. Hence, such theories should have an operational description in terms of such experiments.

The OPT framework is based on the graphical language of circuits, describing experiments that can be performed in a laboratory with physical systems connecting together physical processes, which are denoted as wires and boxes respectively.  The systems/wires are labelled with a \emph{type} denoted $\mathrm{A}$, $\mathrm{B}$, $\mathrm{C}$, \ldots{}{. F}or example, the type given to a quantum system is the dimension of the Hilbert space describing the system. The processes/boxes are then viewed as transformations with some input and output systems/wires{. For instance}, in quantum theory these correspond to quantum instruments. We now give a brief introduction to the important concepts in this formalism.

\subsection{States, transformations, and effects}

A fundamental tenant of the OPT framework is composition of systems and physical processes. Given two systems $\mathrm{A}$ and $\mathrm{B}$, they can be combined into a composite system, denoted by $\mathrm{A}\otimes\mathrm{B}$.
Physical processes can be composed to build circuits, such as

\begin{equation}\label{circuit} \begin{aligned}\Qcircuit @C=1em @R=.7em @!R { & \multiprepareC{1}{\rho} & \qw \poloFantasmaCn{\rA} & \gate{\cA} & \qw \poloFantasmaCn{\rA'} & \gate{\cA'} & \qw \poloFantasmaCn{\rA''} &\measureD{a} \\ & \pureghost{\rho} & \qw \poloFantasmaCn{\rB} & \gate{\cB} & \qw \poloFantasmaCn{\rB'} &\qw &\qw &\measureD{b} }\end{aligned}~. \end{equation}

Processes with no inputs (such as $\rho$ in the above diagram) are called \emph{states}, those with no outputs (such as $a$ and $b$) are called \emph{effects} and, those with both inputs and outputs (such as $\cA$, $\cA'$, $\cB$) {are called} transformations. We define:

\begin{enumerate}
\item$\mathsf{St}\left(\mathrm{A}\right)$ as the set of states of system $\mathrm{A}$,
\item $\mathsf{Eff}\left(\mathrm{A}\right)$ as the set of effects on $\mathrm{A}$,
\item $\mathsf{Transf}\left(\mathrm{A},\mathrm{B}\right)$ as the set of transformations
from $\mathrm{A}$ to $\mathrm{B}$, and $\mathsf{Transf}\left(\mathrm{A}\right)$
as the set of transformations from $\mathrm{A}$ to $\mathrm{A}$,
\item $\mathcal{B}\circ\mathcal{A}$ (or $\mathcal{B}\mathcal{A}$, for
short) as the sequential composition of two transformations $\mathcal{A}$
and $\mathcal{B}$, with the input of $\mathcal{B}$ matching the
output of $\mathcal{A}$,
\item $\mathcal{A}\otimes\mathcal{B}$ as the parallel composition (or tensor
product) of the transformations $\mathcal{A}$ and $\mathcal{B}$.
\end{enumerate}

OPTs include a particular system, the trivial system $\mathrm{I}$,
representing the lack of input or output for a particular device.

Hence, states (resp.\ effects) are transformations with the trivial system as input (resp.\ output). Circuits with no external wires, like the circuit in Equation~\eqref{circuit}, are called scalars and are associated with probabilities. We will often use the notation $\left(a\middle|\rho\right)$ to denote the
circuit\[ \left(a\middle|\rho\right)~:=\!\!\!\!\begin{aligned}\Qcircuit @C=1em @R=.7em @!R { & \prepareC{\rho} & \qw \poloFantasmaCn{\rA} &\measureD{a}}\end{aligned}~, \]and
of the notation $\left(a\middle|\mathcal{C}\middle|\rho\right)$ to denote
the circuit\[ \left(a\middle|\cC\middle|\rho\right)~:=\!\!\!\!\begin{aligned}\Qcircuit @C=1em @R=.7em @!R { & \prepareC{\rho} & \qw \poloFantasmaCn{\rA} &\gate{\cC} &\qw \poloFantasmaCn{\rB} &\measureD{a}}\end{aligned}~. \]

The fact that scalars are probabilities and so are real numbers induces a notion of a sum of
transformations, so that the sets $\mathsf{St}\left(\mathrm{A}\right)$,
$\mathsf{Transf}\left(\mathrm{A},\mathrm{B}\right)$, and $\mathsf{Eff}\left(\mathrm{A}\right)$
become spanning sets of real vector spaces, denoted by $\mathsf{St}_{\mathbb{R}}\left(\mathrm{A}\right)$,
$\mathsf{Transf}_{\mathbb{R}}\left(\mathrm{A},\mathrm{B}\right)$,
and $\mathsf{Eff}_{\mathbb{R}}\left(\mathrm{A}\right)$. In this work
we will restrict our attention to \emph{finite} systems, i.e.,\ systems
for which the vector space spanned by states is finite-dimensional
for all systems. Operationally this assumption means that one need not perform an infinite number of distinct experiments to fully characterise a state. Restricting ourselves to non-negative real numbers,
we have the convex cone of states and of effects, denoted by $\mathsf{St}_{+}\left(\mathrm{A}\right)$
and $\mathsf{Eff}_{+}\left(\mathrm{A}\right)$ respectively. We moreover make the assumption that the set of states is closed. Operationally this is justified by the fact that up to any experimental error a state space is indistinguishable from its closure.

The composition of states and effects leads naturally to  a norm. This is defined, for states $\rho$ as $\left\Vert \rho\right\Vert :=\sup_{a\in\mathsf{Eff}\left(\mathrm{A}\right)}\left(a\middle|\rho\right)$,
and similarly for effects $a$ as $\left\Vert a\right\Vert :=\sup_{\rho\in\mathsf{St}\left(\mathrm{A}\right)}\left(a\middle|\rho\right)$.
The set of normalised states (resp.\ effects) of system $\mathrm{A}$ is denoted by $\mathsf{St}_{1}\left(\mathrm{A}\right)$ (resp.\ $\mathsf{Eff}_{1}\left(\mathrm{A}\right)$).

Transformations are characterised by their action on states of composite
systems: if $\mathcal{A},\mathcal{A}'\in\mathsf{Transf}\left(\mathrm{A},\mathrm{B}\right)$,
we have that $\mathcal{A}=\mathcal{A}'$ if and only if

\begin{equation}\label{eq:equality transformations}
\begin{aligned}\Qcircuit @C=1em @R=.7em @!R { & \multiprepareC{1}{\rho} & \qw \poloFantasmaCn{\rA} & \gate{\cA} & \qw \poloFantasmaCn{\rB} & \qw \\ & \pureghost{\rho} & \qw \poloFantasmaCn{\rS} & \qw &\qw &\qw }\end{aligned}~=\!\!\!\!\begin{aligned}\Qcircuit @C=1em @R=.7em @!R { & \multiprepareC{1}{\rho} & \qw \poloFantasmaCn{\rA} & \gate{\cA'} & \qw \poloFantasmaCn{\rB} & \qw \\ & \pureghost{\rho} & \qw \poloFantasmaCn{\rS} & \qw &\qw &\qw }\end{aligned}~,
\end{equation}for every system $\mathrm{S}$ and every state $\rho\in\mathsf{St}\left(\mathrm{A}\otimes\mathrm{S}\right)$. However it follows that \cite{Chiribella-purification} effects (resp. states) are completely defined by their action on states (resp. effects) of a single system.

Equality on states of the single system $\mathrm{A}$ is, in general, not enough
to discriminate between $\mathcal{A}$ and $\mathcal{A}'$, as is the case for quantum theory over real Hilbert spaces \cite{Wootters-real}.
However, for the scope of the present article, which focuses on single{-}system properties, we {often} concern ourselves with equality on single system.
\begin{Definition}
Two transformations $\mathcal{A},\mathcal{A}'\in\mathsf{Transf}\left(\mathrm{A},\mathrm{B}\right)$
are \emph{equal on single system}, denoted by $\mathcal{A}\doteq\mathcal{A}'$,
if $\mathcal{A}\rho=\mathcal{A}'\rho$ for all states $\rho\in\mathsf{St}\left(\mathrm{A}\right)$.
\end{Definition}

\subsection{Tests and channels}

In general, the boxes corresponding to physical processes come equipped with classical pointers. When used in an experiment, the final position of the a given pointer indicates the particular process which occurred for that box in that run. In general, this procedure can be non-deterministic. These non-deterministic
processes are described by \emph{tests} \cite{Chiribella-purification,QuantumFromPrinciples}:
a test from $\mathrm{A}$ to $\mathrm{B}$ is a collection of transformations
$\left\{ \mathcal{C}_{i}\right\} _{i\in\mathsf{X}}$ from $\mathrm{A}$
to $\mathrm{B}$, where $\mathsf{X}$ is the set of outcomes. If $\mathrm{A}$
(resp.\ $\mathrm{B}$) is the trivial system, the test is called
a \emph{preparation-test} (resp.\ \emph{observation-test}). If the
set of outcomes $\mathsf{X}$ has a single element, we say that the
test is \emph{deterministic}, because only one transformation can
occur. Deterministic transformations will be called \emph{channels}.

A channel $\mathcal{U}$ from $\mathrm{A}$ to $\mathrm{B}$ is \emph{reversible}
if there exists another channel $\mathcal{U}^{-1}$ from $\mathrm{B}$
to $\mathrm{A}$ such that $\mathcal{U}^{-1}\mathcal{U}=\mathcal{I}_{\mathrm{A}}$
and $\mathcal{U}\mathcal{U}^{-1}=\mathcal{I}_{\mathrm{B}}$, where
$\mathcal{I}_{\mathrm{S}}$ is the identity transformation on system $\mathrm{S}$.
If there exists a reversible channel transforming $\mathrm{A}$ into
$\mathrm{B}$, we say that $\mathrm{A}$ and $\mathrm{B}$ are \emph{operationally
equivalent}, denoted as $\mathrm{A}\simeq\mathrm{B}$. The composition
of systems is required to be \emph{symmetric}, meaning that $\mathrm{A}\otimes\mathrm{B}\simeq\mathrm{B}\otimes\mathrm{A}$.
Physically, this means that for every pair of systems there exists
a reversible channel swapping them. A state $\chi$ is called \emph{invariant} if $\mathcal{U}\chi=\chi$
for all reversible channels $\mathcal{U}$.

A particularly useful class of observation-tests allows for the following.
\begin{Definition}
The states $\left\{ \rho_{i}\right\} _{i\in\mathsf{X}}$ are called
\emph{perfectly distinguishable} if there exists an observation-test
$\left\{ a_{i}\right\} _{i\in\mathsf{X}}$ such that $\left(a_{i}\middle|\rho_{j}\right)=\delta_{ij}$
for all $i,j\in\mathsf{X}$.

Moreover, if there is \emph{no} other state $\rho_{0}$ such that
the states $\left\{ \rho_{i}\right\} _{i\in\mathsf{X}}\cup\left\{ \rho_{0}\right\} $
are perfectly distinguishable, the set $\left\{ \rho_{i}\right\} _{i\in\mathsf{X}}$
is said \emph{maximal}.
\end{Definition}

\subsection{Pure transformations}

There are various different ways to define pure transformations, for example in terms of resources \cite{Horodecki-Oppenheim,Nicole,Chiribella-Scandolo-entanglement,TowardsThermo,Purity} or ``side information'' \cite{QuantumFromPrinciples,Selby-leaks}. Informally pure transformations correspond to an experimenter having maximal control of or information about a process.
Here, we formalise this notion
by defining
the notion of a \emph{coarse-graining}
\cite{Chiribella-purification}. Coarse-graining is the operation of joining two or more outcomes of
a test into a single outcome. More precisely, a test $\left\{ \mathcal{C}_{i}\right\} _{i\in\mathsf{X}}$
is a \emph{coarse-graining} of the test $\left\{ \mathcal{D}_{j}\right\} _{j\in\mathsf{Y}}$
if there is a partition $\left\{ \mathsf{Y}_{i}\right\} _{i\in\mathsf{X}}$
of $\mathsf{Y}$ such that, for all $i\in\mathsf{X}$
\[
\mathcal{C}_{i}=\sum_{j\in\mathsf{Y}_{i}}\mathcal{D}_{j}
\]

In this case, we say that the test $\left\{ \mathcal{D}_{j}\right\} _{j\in\mathsf{Y}}$
is a \emph{refinement} of the test $\left\{ \mathcal{C}_{i}\right\} _{i\in\mathsf{X}}$,
and that the transformations $\left\{ \mathcal{D}_{j}\right\} _{j\in\mathsf{Y}_{i}}$
are a refinement of the transformation $\mathcal{C}_{i}$. A transformation
$\mathcal{C}\in\mathsf{Transf}\left(\mathrm{A},\mathrm{B}\right)$
is \emph{pure} if it has only trivial refinements, namely refinements
$\left\{ \mathcal{D}_{j}\right\} $ of the form $\mathcal{D}_{j}=p_{j}\mathcal{C}$,
where $\left\{ p_{j}\right\} $ is a probability distribution.
We denote the sets of pure transformations,
pure states, and pure effects as $\mathsf{PurTransf}\left({\rm A,{\rm B}}\right)$,
$\mathsf{PurSt}\left({\rm A}\right)$, and $\mathsf{PurEff}\left({\rm A}\right)$
respectively. Similarly, $\mathsf{PurSt}_{1}\left({\rm A}\right)$,
and $\mathsf{PurEff}_{1}\left({\rm A}\right)$ denote \emph{normalised}
pure states and effects respectively. Non-pure states are
called \emph{mixed}.
\begin{Definition}
Let $\rho\in\mathsf{St}_{1}\left(\mathrm{A}\right)$. A normalised
state $\sigma$ is \emph{contained} in $\rho$ if we can write $\rho=p\sigma+\left(1-p\right)\tau$,
where $p\in\left(0,1\right]$ and $\tau$ is another normalised state.
\end{Definition}
{
Clearly, no states are contained in a pure state. On the other edge of the spectrum we have complete states.
\begin{Definition}
A state $\omega\in\mathsf{St}_{1}\left(\mathrm{A}\right)$ is \emph{complete}
if every state is contained in it.
\end{Definition}}
\vspace{-12pt}

\begin{Definition}
We say that two transformations $\mathcal{A},\mathcal{A}'\in\mathsf{Transf}\left(\mathrm{A},\mathrm{B}\right)$
are \emph{equal upon input} of the state $\rho\in\mathsf{St}_{1}\left(\mathrm{A}\right)$
if $\mathcal{A}\sigma=\mathcal{A}'\sigma$ for every state $\sigma$
contained in $\rho$. In this case we will write $\mathcal{A}=_{\rho}\mathcal{A}'$.
\end{Definition}

\subsection{Causality}

A natural requirement of a physical theory is that it is \emph{causal}, that is, no signals can be sent from the future to the past. In the OPT framework this is formalised as follows:
\begin{ax}[Causality \cite{Chiribella-purification,QuantumFromPrinciples}]
The probability that a transformation occurs is independent of the
choice of tests performed on its output.
\end{ax}
Causality is equivalent to the requirement that, for every system $\mathrm{A}$,
there exists a unique deterministic effect $u_{\mathrm{A}}$ on $\mathrm{A}$
(or simply $u$, when no ambiguity can arise) \cite{Chiribella-purification}.
Owing to the uniqueness of the deterministic effect, the marginals
of a bipartite state can be uniquely defined as:
\[
\begin{aligned}\Qcircuit @C=1em @R=.7em @!R { &\prepareC{\rho_\rA} & \qw \poloFantasmaCn{\rA} &\qw } \end{aligned} ~:= \!\!\!\! \begin{aligned}\Qcircuit @C=1em @R=.7em @!R { &\multiprepareC{1}{\rho_{\rA\rB}} & \qw \poloFantasmaCn{\rA} &\qw \\ & \pureghost{\rho_{\rA\rB}} & \qw \poloFantasmaCn{\rB} & \measureD{u} }\end{aligned}~,
\]

Moreover, this uniqueness forbids the ability to signal \cite{Chiribella-purification,Coecke-no-signalling}. We will denote
by $\mathrm{Tr}_{\mathrm{B}}\rho_{\mathrm{AB}}$ the marginal on system
$\mathrm{A}$, in analogy with the notation used in the quantum case. We will stick to the notation $\mathrm{Tr}$ in formulas where the
deterministic effect is applied directly to a state, e.g.,\ $\mathrm{Tr}\:\rho:=\left(u\middle|\rho\right)$.

In a causal theory it is easy to see that the norm of a state takes
the form $\left\Vert \rho\right\Vert =\mathrm{Tr}\:\rho$, and that
a state can be prepared deterministically if and only if it is normalised.

\section{Higher-order interference} \label{higher-order interference}

The definition of higher-order interference we shall present in this section takes its motivation from the set-up of multi-slit interference experiments. In such experiments a particle passes through slits in a physical barrier and is detected at a screen. By repeating the experiment many times, one builds up a pattern on the screen. To determine if this experiment exhibits interference one compares this pattern to those produced when certain subsets of the slits are blocked. In quantum theory, for example, the two-slit experiment exhibits interference as the pattern formed with both slits open is not equal to the sum of the one-slit patterns.

Consider the state of the particle just before it passes through the slits. For every slit, there should exist states such that the particle is definitely found at that slit, if measured. Mathematically, this means that there is a face \cite{Barnum-interference} of the state space, such that all states in this face give unit probability for the ``yes'' outcome of the two-outcome measurement ``is the particle at this slit?''. Recall that a face is a convex set with the property that if $px+\left(1-p\right)y$, for $0 \leq p \leq 1$, is an element then $x$ and $y$ are also elements. These faces will be labelled $F_i$, one for each of the $n$ slits $i\in\left\{1,\ldots, n\right\}$. As the slits should be perfectly distinguishable, the faces associated with each slit should be perfectly distinguishable, or orthogonal. One can additionally ask coarse-grained questions of the form ``Is the particle found among a certain subset of slits, rather than somewhere else?''. The set of states that give outcome ``yes'' with probability one must contain all the faces associated with each slit in the subset. Hence the face associated with the subset of slits $\mathsf{I}\subseteq\left\{1,\ldots, n\right\}$ is the smallest face containing each face in this subset $F_{\mathsf{I}}:=\bigvee_{i\in \mathsf{I}} F_i$, where the operation $\bigvee$ is the least upper bound of the lattice of faces where the ordering is provided by subset inclusion of one face within another. The face $F_{\mathsf{I}}$ contains all those states which can be found among the slits contained in $\mathsf{I}$. The experiment is ``complete'' if all states in the state space (of a given system $\mathrm{A}$) can be found among some subset of slits. That is, if $F_{12\cdots n}=\mathsf{St}\left(\mathrm{A}\right)$.

An $n$-slit experiment requires a system that has $n$ orthogonal faces $F_i$, with $i\in\left\{1,\ldots,n\right\}$. Consider an effect $E$ associated with finding a particle at a particular point on the screen. We now formally define an $n$-slit experiment.

\begin{Definition}\label{def:faces interference}
An $n$-slit experiment is a collection of effects $e_{\mathsf{I}}$, where $ \mathsf{I}\subseteq \left\{1,\ldots,n\right\}${,} such that $$ \begin{aligned}
\left(e_{\mathsf{I}}\middle|\rho\right)&=\left(E\middle|\rho\right),\hspace{3mm} \forall \rho\in F_{\mathsf{I}}, \text{ and} \\
\left(e_{\mathsf{I}}\middle|\rho\right)&= 0,\hspace{3mm} \forall \rho \textrm{ where }\rho\perp F_{\mathsf{I}}.
\end{aligned}
$$
\end{Definition}
The effects introduced in the above definition arise from the conjunction of blocking off the slits $\left\{1,\ldots,n\right\}\setminus \mathsf{I}$ and applying the effect $E$. If the particle was prepared in a state such that it would be unaffected by the blocking of the slits (i.e.,\ $\rho\in F_{\mathsf{I}}$) then we should have $\left(e_{\mathsf{I}}\middle|\rho\right)=\left(E\middle|{\rho}\right)$. If instead the particle is prepared in a state which is guaranteed to be blocked (i.e.,\ $\rho'\perp F_{\mathsf{I}}$) then the particle should have no probability of being detected at the screen, i.e., $\left(e_{\mathsf{I}}\middle|\rho'\right)=0$.

The relevant quantities for the existence of various orders of interference are \cite{Barnum-1,Ududec-3slits,Sorkin1,Lee-Selby-interference}:
\begin{eqnarray}
& &I_1:=\left(E\middle|\rho\right),\\
& &I_2:=\left(E\middle|\rho\right)-\left(e_1\middle|\rho\right)-\left(e_2\middle|\rho\right),\\
& &I_3:=\left(E\middle|\rho\right)-\left(e_{12}\middle|\rho\right)-\left(e_{23}\middle|\rho\right)-\left(e_{31}\middle|\rho\right)+\left(e_1\middle|\rho\right)+\left(e_2\middle|\rho\right)+\left(e_3\middle|\rho\right),\\
& &I_n:= \sum_{\varnothing\neq \mathsf{I}\subseteq\left\{1,\ldots,n\right\}}\left(-1\right)^{n-\left|\mathsf{I}\right|}\left(e_{\mathsf{I}}\middle|\rho\right), \label{HOI}
\end{eqnarray}
for some state $\rho${,} and defining $e_{\left\{1,\ldots,n\right\}}:=E$.

\begin{Definition}
A theory has $n$-th order interference if there exists a state $\rho$ and an effect $E$ such that $I_n\neq0$.
\end{Definition}

In a slightly different formal setting, it was shown in \cite{Sorkin1} that $I_n = 0 \implies I_{n+1}=0$, so if there is no $n$th order interference, there will be no $\left(n+1\right)$th order interference; the argument of \cite{Sorkin1} applies here.

It should be noted that there appears to be a lot of freedom in choosing a set of effects $\left\{e_{\mathsf{I}}\right\}$ to test for the existence of higher-order interference. Indeed, in arbitrary generalised theories this appears to be the case \cite{Lee-Selby-interference}. However, it is natural to ask whether there exists physical transformations $T_{\mathsf{I}}$ in the theory which correspond to leaving the subset of slits $\mathsf{I}$ open and blocking the rest. Hence  a unique $e_{\mathsf{I}}$ is assigned to each fixed $E$
defined as $e_{\mathsf{I}}=ET_{\mathsf{I}}$. Ruling out the existence of higher-order interference then reduces to proving certain properties of the $T_{\mathsf{I}}$. This will turn out to be the case in sharp theories with purification.

\section{Sharp theories with purification} \label{sharp theories with purification}

In this section we present the definition and important properties of sharp theories with purification{. They} were originally introduced in \cite{Diagonalization,TowardsThermo,Purity}
for the analysis of the foundations of thermodynamics and statistical
mechanics.

{Sharp theories with purification are causal theories defined by three
axioms. The first axiom\textemdash Purity Preservation\textemdash states
that no information can leak when two pure transformations are composed:}
\begin{ax}[Purity Preservation \cite{Scandolo14}]
Sequential and parallel compositions of pure transformations yield
pure transformations.
\end{ax}
The second axiom\textemdash Pure Sharpness\textemdash guarantees that
every system possesses at least one elementary property.
\begin{ax}[Pure Sharpness \cite{Diagonalization}]
For every system there exists at least one pure effect occurring
with unit probability on some state.
\end{ax}
These axioms are satisfied by both classical and quantum theory. Our third axiom---Purification---signals the departure
from classicality, and characterises when a physical theory admits a level of description where all deterministic processes are pure and reversible.

Given a normalised state $\rho_{\mathrm{A}}\in\mathsf{St}_{1}\left(\mathrm{A}\right)$,
a normalised pure state $\Psi\in\mathsf{PurSt}_{1}\left(\mathrm{A}\otimes\mathrm{B}\right)$
is a \emph{purification} of $\rho_{\mathrm{A}}$ if
\[ \begin{aligned}\Qcircuit @C=1em @R=.7em @!R { &\multiprepareC{1}{\Psi} & \qw \poloFantasmaCn{\rA} &\qw \\ & \pureghost{\Psi} & \qw \poloFantasmaCn{\rB} & \measureD{u} } \end{aligned} ~ =\!\!\!\! \begin{aligned}\Qcircuit @C=1em @R=.7em @!R { &\prepareC{\rho_{\rA}} & \qw \poloFantasmaCn{\rA} &\qw } \end{aligned}~{;} \]
in this case $\mathrm{B}$ is called the \emph{purifying system}. We say that a pure state $\Psi\in\mathsf{PurSt}\left(\mathrm{A}\otimes\mathrm{B}\right)$
is an \emph{essentially unique purification of its marginal $\rho_{\mathrm{A}}$}
\cite{QuantumFromPrinciples} if every other pure state $\Psi'\in\mathsf{PurSt}\left(\mathrm{A}\otimes\mathrm{B}\right)$
satisfying the purification condition must
be of the form\[\begin{aligned}\Qcircuit @C=1em @R=.7em @!R { & \multiprepareC{1}{\Psi'} & \qw \poloFantasmaCn{\rA} & \qw \\ & \pureghost{\Psi'} & \qw \poloFantasmaCn{\rB} & \qw }\end{aligned}~=\!\!\!\! \begin{aligned}\Qcircuit @C=1em @R=.7em @!R { & \multiprepareC{1}{\Psi} & \qw \poloFantasmaCn{\rA} & \qw &\qw &\qw \\ & \pureghost{\Psi} & \qw \poloFantasmaCn{\rB} & \gate{\cU} &\qw \poloFantasmaCn{\rB} &\qw }\end{aligned}~ ,\]for
some reversible channel $\mathcal{U}$.
\begin{ax}[Purification \cite{Chiribella-purification,QuantumFromPrinciples}]
Every state has a purification. Purifications are essentially unique.
\end{ax}
Quantum theory, both on complex and real Hilbert spaces, satisfies
Purification, and also Spekkens' toy model \cite{Disilvestro}. Examples of sharp theories with purification besides quantum theory include fermionic quantum theory \cite{Fermionic1,Fermionic2}, a superselected version of quantum theory known as double{d} quantum theory \cite{Purity}, and a recent extension of classical theory with the
theory of codits \cite{TowardsThermo}.

\subsection{Properties of sharp theories with purifications}

Sharp theories with purifications enjoy some nice properties, which
were mainly derived in Refs.~\cite{Diagonalization,TowardsThermo}. The first property is that every non-trivial system admits perfectly
distinguishable states~\cite{Diagonalization}, and that all maximal
sets of pure states have the same cardinality \cite{TowardsThermo}.
\begin{Proposition}
For every system $\mathrm{A}$ there is a positive integer $d_{\mathrm{A}}$,
called the \emph{dimension} of $\mathrm{A}$, such that all maximal
sets of pure states have $d_{\mathrm{A}}$ elements.
\end{Proposition}

Note that we will omit the subscript $\mathrm{A}$ when the context is clear.

In sharp theories with purification every state can be diagonalised,
i.e.,\ written as a convex combination of perfectly distinguishable
pure states (cf.\ Refs.~\cite{Diagonalization,TowardsThermo}).
\begin{thm}
\label{thm:diagonalisation}Every normalised state $\rho\in\mathsf{St}_{1}\left(\mathrm{A}\right)$
of a non-trivial system can be decomposed as
\[
\rho=\sum_{i=1}^{d}p_{i}\alpha_{i},
\]
where $\left\{ p_{i}\right\} _{i=1}^{d}$ is a probability distribution,
and $\left\{ \alpha_{i}\right\} _{i=1}^{d}$ is a \emph{pure} maximal
set. Moreover, given $\rho$, $\left\{ p_{i}\right\} _{i=1}^{d}$
is unique up to rearrangements.
\end{thm}
Such a decomposition is called a \emph{diagonalisation} of $\rho$,
the $p_{i}$'s are the \emph{eigenvalues} of $\rho$, and the $\alpha_{i}$'s
are the \emph{eigenstates}. Theorem~\ref{thm:diagonalisation} implies
that the eigenvalues of a state are unique, and independent of its
diagonalisation. Sharp theories with purification have a unique invariant
state $\chi$ \cite{Chiribella-purification}, which can be diagonalised as $\chi=\frac{1}{d}\sum_{i=1}^{d}\alpha_{i}$,
where $\left\{ \alpha_{i}\right\} _{i=1}^{d}$ is \emph{any} pure
maximal set \cite{TowardsThermo}.
Furthermore, the diagonalisation result of Theorem~\ref{thm:diagonalisation}
can be extended to every vector in $\mathsf{St}_{\mathbb{R}}\left(\mathrm{A}\right)$,
but here the eigenvalues will be generally real numbers \cite{TowardsThermo}.

One of the most important consequences for this paper of the axioms
defining sharp theories with purification is a duality between normalised
pure states and normalised pure effects.
\begin{thm}[States-effects duality \cite{Diagonalization,TowardsThermo}]
\label{thm:duality}For every system $\mathrm{A}$, there is a bijective
correspondence $\dagger:\mathsf{PurSt}_{1}\left(\mathrm{A}\right)\rightarrow\mathsf{PurEff}_{1}\left(\mathrm{A}\right)$
such that if $\alpha\in\mathsf{PurSt}_{1}\left(\mathrm{A}\right)$,
$\alpha^{\dagger}$ is the \emph{unique} normalised pure effect such
that $\left(\alpha^{\dagger}\middle|\alpha\right)=1$.
Furthermore this bijection can be extended by linearity to an isomorphism
between the vector spaces $\mathsf{St}_{\mathbb{R}}\left(\mathrm{A}\right)$
and $\mathsf{Eff}_{\mathbb{R}}\left(\mathrm{A}\right)$.
\end{thm}
With a little abuse of notation we will use $\dagger$ also to denote
the inverse map $\mathsf{PurEff}_{1}\left(\mathrm{A}\right)\rightarrow\mathsf{PurSt}_{1}\left(\mathrm{A}\right)$,
by which, if $a\in\mathsf{PurEff}_{1}\left(\mathrm{A}\right)$, $a^{\dagger}$
is the unique pure state such that $\left(a\middle|a^{\dagger}\right)=1$.
Pure maximal sets $\left\{ \alpha_{i}\right\} _{i=1}^{d}$ have the
property that $\sum_{i=1}^{d}\alpha_{i}^{\dagger}=u$ \cite{TowardsThermo}.

A diagonalisation result holds for vectors of $\mathsf{Eff}_{\mathbb{R}}\left(\mathrm{A}\right)$ as well \cite{TowardsThermo}: they can be written as $X = \sum_{i=1}^{d}\lambda_i\alpha_i^{\dagger}$, where $\left\{\alpha_i\right\}_{i=1}^d$ is a pure maximal set. Again, the $\lambda_i$'s are uniquely defined given $X$.

Another result that will be made use of in the following sections is the following.
It was shown to hold in Ref.~\cite{TowardsThermo}, and expresses
the possibility of constructing non-disturbing measurements \cite{Chiribella-informational,Wehner-Pfister,Chiribella-Yuan2015}.
\begin{Proposition}
\label{prop:non-disturbing}Given a system $\mathrm{A}$, let $a\in\mathsf{Eff}\left(\mathrm{A}\right)$
be an effect such that $\left(a\middle|\rho\right)=1$, for some $\rho\in\mathsf{St}_{1}\left(\mathrm{A}\right)$.
Then there exists a pure transformation $\mathcal{T}\in\mathsf{PurTransf}\left(\mathrm{A}\right)$
such that $\mathcal{T}=_{\rho}\mathcal{I}$, {with} $\left(u\middle|\mathcal{T}\middle|\sigma\right)\le\left(a\middle|\sigma\right)$,
for every state $\sigma\in\mathsf{St}_{1}\left(\mathrm{A}\right)$.
\end{Proposition}
Note that the pure transformation $\mathcal{T}$ is non-disturbing
on $\rho$ because it acts as the identity on $\rho$ and on all
states contained in it. In other words, whenever we have an effect occurring with unit probability on some state $\rho$,
we can always find a transformation that does not disturb $\rho$
(i.e.,\ a non-disturbing, non-demolition measurement) \cite{TowardsThermo}.

Finally, a property that we will use often is a sort of no-restriction
hypothesis for tests, derived in~\cite{Chiribella-informational} (Corollary 4).
\begin{Proposition}
A collection of transformations $\left\{ \mathcal{A}_{i}\right\} _{i\in\mathsf{X}}$
is a valid test if and only if $\sum_{i\in\mathsf{X}}u\mathcal{A}_{i}=u$.

\noindent
A collection of effects $\left\{ a_{i}\right\} _{i\in\mathsf{X}}$
is a valid observation-test if and only if $\sum_{i\in\mathsf{X}}a_{i}=u$.
\end{Proposition}

\section{Sharp theories with purification have no higher-order interference} \label{proofs}

Here we will show that sharp theories with purification do not exhibit higher-order interference. Our proof strategy will be to show that results of \cite{Barnum-interference}, which rule out the existence of higher-order interference from certain assumptions, hold in sharp theories with purification. To this end, we will first prove
that these theories are self-dual,
and that they admit \emph{pure} orthogonal projectors which satisfy certain properties{, compatible with the setting presented in Section~\ref{higher-order interference}.}

\subsection{Self-duality\label{subsec:Self-duality}}
Now we will prove that sharp theories with purification are self-dual. Recall that a theory is \emph{self-dual} if for every system $\mathrm{A}$
there is an inner product $\left\langle \bullet,\bullet\right\rangle $
on $\mathsf{St}_{\mathbb{R}}\left(\mathrm{A}\right)$ such that $\xi\in\mathsf{St}_{+}\left(\mathrm{A}\right)$
if and only if $\left\langle \xi,{\eta}\right\rangle \geq0$ for every
${\eta}\in\mathsf{St}_{+}\left(\mathrm{A}\right)$. To show that, we need to find a self-dualising inner product on $\mathsf{St}_{\mathbb{R}}\left(\mathrm{A}\right)$
for every system $\mathrm{A}$. The dagger will provide us with a
good candidate. First we need the following lemma.
\begin{Lemma}
\label{lem:normalised effects}Let $a\in\mathsf{Eff}_{1}\left(\mathrm{A}\right)$
be a normalised effect. Then $a$ can be diagonalised as  $a=\sum_{i\in\mathsf{I}}\alpha_{i}^{\dagger}+\sum_{j\in\mathsf{J}}\lambda_{j}\alpha_{j}^{\dagger}$, where $\mathsf{I}$ is a \emph{non-empty} subset of $\left\lbrace 1,\ldots,d\right\rbrace  $, and  $\mathsf{J}$ is  a (possibly empty) subset of the complement of $\mathsf{I}$, and $ \lambda_{j}\in\left( 0,1\right)  $ for every $ j\in\mathsf{J} $.
\end{Lemma}
\begin{proof}
{W}e know that every effect $a$ can be written as $a=\sum_{i=1}^{r}\lambda_{i}\alpha_{i}^{\dagger}$,
where $r\leq d$, the {pure} states $\left\{ \alpha_{i}\right\} _{i=1}^{r}$
are perfectly distinguishable, and for every $i\in\left\{ 1,\ldots,r\right\} $,
$\lambda_{i}\in\left(0,1\right]$. Since the state space is closed,
and $a$ is normalised, then there exist{s} a (normalised) state $\rho$
such that $\left(a\middle|\rho\right)=1$. {O}ne has
\[
1=\left(a\middle|{\rho}\right)=\sum_{i=1}^{r}\lambda_{i}\left(\alpha_{i}^{\dagger}\middle|{\rho}\right).
\]
Now, $\left(\alpha_{i}^{\dagger}\middle|{\rho}\right)\geq0$, and $\sum_{i=1}^{r}\left(\alpha_{i}^{\dagger}\middle|{\rho}\right)\leq1$
because
\[
\sum_{i=1}^{r}\left(\alpha_{i}^{\dagger}\middle|{\rho}\right)\leq\sum_{i=1}^{d}\left(\alpha_{i}^{\dagger}\middle|{\rho}\right)=\mathrm{Tr}\:{\rho}=1,
\]
where we have used the fact that $\sum_{i=1}^{d}\alpha_{i}^{\dagger}=u$. Then $\sum_{i=1}^{r}\lambda_{i}\left(\alpha_{i}^{\dagger}\middle|{\rho}\right)\leq\lambda_{\mathrm{max}}$,
where $\lambda_{\mathrm{max}}$ is the maximum of the $\lambda_{i}$'s.
Therefore, $\lambda_{\mathrm{max}}\geq1$, which implies $\lambda_{\mathrm{max}}=1$.
Now, the condition
\[
\sum_{i=1}^{r}\lambda_{i}\left(\alpha_{i}^{\dagger}\middle|{\rho}\right)=\lambda_{\mathrm{max}}
\]
means that $\lambda_{i}=\lambda_{\mathrm{max}}=1$ for every $i$ such that $\left(\alpha_{i}^{\dagger}\middle|{\rho}\right)>0$. This means that there always exists at least one eigenstate with eigenvalue 1, but in general there may be some eigenstates with eigenvalues strictly less than 1.
\end{proof}
We can use
this result to prove the following.
\begin{Lemma}
\label{lem:inner product}For every system $\mathrm{A}$, the map
\[
\left\langle \xi,\eta\right\rangle :=\left(\xi^{\dagger}\middle|\eta\right),
\]
for every $\xi,\eta\in\mathsf{St}_{\mathbb{R}}\left(\mathrm{A}\right)$
is an inner product on $\mathsf{St}_{\mathbb{R}}\left(\mathrm{A}\right)$.
\end{Lemma}
\begin{proof}
The map $\left\langle \bullet,\bullet\right\rangle $ is clearly bilinear
by construction, because the dagger is also linear. Let us show that
it is positive-definite. Take a non-null vector $\xi\in\mathsf{St}_{\mathbb{R}}\left(\mathrm{A}\right)$,
and diagonalise it as $\xi=\sum_{i=1}^{d}x_{i}\alpha_{i}$. Then
\[
\left\langle \xi,\xi\right\rangle =\left(\xi^{\dagger}\middle|\xi\right)=\sum_{i,j=1}^{d}x_{i}x_{j}\left(\alpha_{i}^{\dagger}\middle|\alpha_{j}\right)=\sum_{i=1}^{d}x_{i}^{2}>0,
\]
where we have used the fact that for perfectly distinguishable pure
states $\left(\alpha_{i}^{\dagger}\middle|\alpha_{j}\right)=\delta_{ij}$
\cite{TowardsThermo}.

The hard part is to prove that this bilinear map is symmetric,
namely $\left\langle \xi,\eta\right\rangle =\left\langle \eta,\xi\right\rangle $,
for every $\xi,\eta\in\mathsf{St}_{\mathbb{R}}\left(\mathrm{A}\right)$.
Let us define a new (double) dagger $\ddagger$. The double dagger
of a normalised state $\rho$ is an effect $\rho^{\ddagger}$ whose
action on normalised states $\sigma$ is defined as
\begin{equation}
\left(\rho^{\ddagger}\middle|\sigma\right):=\left(\sigma^{\dagger}\middle|\rho\right),\label{eq:double dagger}
\end{equation}
where $\dagger$ is the dagger of Theorem~\ref{thm:duality}. Note
that {E}quation~\eqref{eq:double dagger} is enough to characterise $\rho^{\ddagger}$
completely, and it guarantees that $\rho^{\ddagger}$ is a mathematically well-defined
effect, because it is linear and $\left(\sigma^{\dagger}\middle|\rho\right)\in\left[0,1\right]$.
Consider now $\rho$ {and $ \sigma $} to be a normalised \emph{pure} state ${\psi}$. Then $\left({\psi}^{\ddagger}\middle|{\psi}\right)=\left({\psi}^{\dagger}\middle|{\psi}\right)=1$,
this means that $\psi^{\ddagger}$ is normalised. 
By Lemma~\ref{lem:normalised effects},
${\psi}^{\ddagger}$ is of the form ${\psi}^{\ddagger}=\sum_{i\in\mathsf{I}}\alpha_{i}^{\dagger}+\sum_{j\in\mathsf{J}}\lambda_{j}\alpha_{j}^{\dagger}$,
where the pure states $\left\{ \alpha_{i}\right\}_{i\in\mathsf{I}}\cup\left\{ \alpha_{j}\right\}_{j\in\mathsf{J}}$
are perfectly distinguishable. Note that  ${\psi}^{\ddagger}$ is pure if and only if $\left| \mathsf{I}\right| =1 $, and $ \mathsf{J}=\varnothing $.
Let us evaluate ${\psi}^{\ddagger}$
{on} $\chi$:
\begin{equation}
\left({\psi}^{\ddagger}\middle|\chi\right)=\left(\chi^{\dagger}\middle|{\psi}\right)=\frac{1}{d}\mathrm{Tr}\:{\psi}=\frac{1}{d},\label{eq:double dagger chi}
\end{equation}
as prescribed by {E}quation~\eqref{eq:double dagger}. Now, since ${\psi}^{\ddagger}=\sum_{i\in\mathsf{I}}\alpha_{i}^{\dagger}+\sum_{j\in\mathsf{J}}\lambda_{j}\alpha_{j}^{\dagger}$,
we have
\begin{equation}
\left({\psi}^{\ddagger}\middle|\chi\right)=\sum_{i\in\mathsf{I}}\left( \alpha_{i}^{\dagger}\middle|\chi\right) +\sum_{j\in\mathsf{J}}\lambda_{j}\left( \alpha_{j}^{\dagger}\middle|\chi\right) =\frac{1}{d}\left( \left| \mathsf{I}\right| +\sum_{j\in\mathsf{J}}\lambda_{j}\right) ,\label{eq:double dagger chi 2}
\end{equation}
because $\left(\alpha_{i}^{\dagger}\middle|\chi\right)=\frac{1}{d}$ for
every $i$ \cite{TowardsThermo}. Since $ \left| \mathsf{I}\right|\geq1 $ and $\sum_{j\in\mathsf{J}}\lambda_{j}>0  $, a comparison between Equations~\eqref{eq:double dagger chi}
and \eqref{eq:double dagger chi 2} shows that it must be $\left| \mathsf{I}\right| =1 $ and $ \mathsf{J}=\varnothing $. This
means that ${\psi}^{\ddagger}$ is a (physical) pure effect, whence ${\psi}^{\ddagger}={\psi}^{\dagger}$ by Theorem~\ref{thm:duality}.
Now we can show that the double dagger $\ddagger$ actually coincides
with the dagger of Theorem~\ref{thm:duality}. Indeed, given a state
$\rho$, diagonalise it as $\rho=\sum_{i=1}^{d}p_{i}\alpha_{i}$.
One can easily show that the double dagger of {E}quation~\eqref{eq:double dagger}
is linear, so we have $\rho^{\ddagger}=\sum_{i=1}^{d}p_{i}\alpha_{i}^{\ddagger}$,
but we have just proved that $\alpha_{i}^{\ddagger}=\alpha_{i}^{\dagger}$
for pure states, so $\rho^{\ddagger}=\sum_{i=1}^{d}p_{i}\alpha_{i}^{\dagger}=\rho^{\dagger}$.
This means that $\ddagger=\dagger$, and that {E}quation~\eqref{eq:double dagger}
is nothing but a redefinition of the usual dagger. This means for
every normalised states we have
\begin{equation}
\left(\rho^{\dagger}\middle|\sigma\right)=\left(\sigma^{\dagger}\middle|\rho\right),\label{eq:dagger symmetric}
\end{equation}
and this extends linearly to all vectors $\xi,\eta\in\mathsf{St}_{\mathbb{R}}\left(\mathrm{A}\right)$.
We have proved that $\left\langle \bullet,\bullet\right\rangle $
is symmetric, and this concludes the proof.
\end{proof}
Note that the above result immediately yields the ``symmetry of transition probabilities'' as defined in Ref.~\cite{Alfsen-Shultz,Barnum-thermo-QPL}.

Now we prove that this inner product is invariant
under reversible transformations.
\begin{Proposition}
\label{prop:invariance reversible}For every $\xi,\eta\in\mathsf{St}_{\mathbb{R}}\left(\mathrm{A}\right)$
and every reversible channel $\mathcal{U}$ one has
\[
\left\langle \mathcal{U}\xi,\mathcal{U}\eta\right\rangle =\left\langle \xi,\eta\right\rangle .
\]
\end{Proposition}
\begin{proof}
To prove the statement, let us first prove that for a normalised pure
state $\alpha$ one has $\left(\mathcal{U}\alpha\right)^{\dagger}=\alpha^{\dagger}\mathcal{U}^{-1}$,
for every reversible channel $\mathcal{U}$. $\alpha^{\dagger}\mathcal{U}^{-1}$
is a pure effect and one has $\left(\alpha^{\dagger}\mathcal{U}^{-1}\middle|\mathcal{U}\alpha\right)=\left(\alpha^{\dagger}\middle|\alpha\right)=1$.
By the uniqueness of the dagger for normalised pure states, $\alpha^{\dagger}\mathcal{U}^{-1}=\left(\mathcal{U}\alpha\right)^{\dagger}$.
This can be extended by linearity to all vectors $\xi$ in $\mathsf{St}_{\mathbb{R}}\left(\mathrm{A}\right)$,
so $\left(\mathcal{U}\xi\right)^{\dagger}=\xi^{\dagger}\mathcal{U}^{-1}$.
Therefore, when we compute $\left\langle \mathcal{U}\xi,\mathcal{U}\eta\right\rangle $,
we have
\[
\left\langle \mathcal{U}\xi,\mathcal{U}\eta\right\rangle =\left(\xi^{\dagger}\middle|\mathcal{U}^{-1}\mathcal{U}\middle|\eta\right)=\left(\xi^{\dagger}\middle|\eta\right)=\left\langle \xi,\eta\right\rangle .
\]
\end{proof}
The fact that $\left\langle \bullet,\bullet\right\rangle $ is {an inner}
product allows us to define an additional norm {in} sharp theories with
purification: if $\xi\in\mathsf{St}_{\mathbb{R}}\left(\mathrm{A}\right)$,
define the \emph{dagger norm} as
\[
\left\Vert \xi\right\Vert _{\dagger}:=\sqrt{\left\langle \xi,\xi\right\rangle }.
\]

See Appendix~\ref{subsec:operational vs dagger} for an extended discussion on the properties of this norm.

Now we are ready to state the core of this subsection.
\begin{Proposition}
Sharp theories with purification are self-dual.
\end{Proposition}
\begin{proof}
Given a system $\mathrm{A}$, we need to prove that $\xi\in\mathsf{St}_{\mathbb{R}}\left(\mathrm{A}\right)$
is in $\mathsf{St}_{+}\left(\mathrm{A}\right)$ if and only if $\left\langle \xi,\eta\right\rangle \geq0$
for all $\eta\in\mathsf{St}_{+}\left(\mathrm{A}\right)$. Note that
$\xi\in\mathsf{St}_{+}\left(\mathrm{A}\right)$ if and only if it
can be diagonalised as $\xi=\sum_{i=1}^{d}x_{i}\alpha_{i}$, where
the $x_{i}$'s are all non-negative.

Necessity. Suppose $\xi\in\mathsf{St}_{+}\left(\mathrm{A}\right)$,
and take any $\eta\in\mathsf{St}_{+}\left(\mathrm{A}\right)$, diagonalised
as $\eta=\sum_{i=1}^{d}y_{i}\beta_{i}$. Then we have
\[
\left\langle \xi,\eta\right\rangle =\sum_{i,j=1}^{d}x_{i}y_{j}\left(\alpha_{i}^{\dagger}\middle|\beta_{j}\right)\geq0
\]
because all the terms $x_{i}$, $y_{j}$, and $\left(\alpha_{i}^{\dagger}\middle|\beta_{j}\right)$
are non-negative.

Sufficiency. Take $\xi\in\mathsf{St}_{\mathbb{R}}\left(\mathrm{A}\right)$,
and assume that $\left\langle \xi,\eta\right\rangle \geq0$ for all
$\eta\in\mathsf{St}_{+}\left(\mathrm{A}\right)$. Assume $\xi$ is
diagonalised as $\xi=\sum_{i=1}^{d}x_{i}\alpha_{i}$, where the $x_{i}$'s
are generic real numbers. We wish to prove that all the $x_{i}$'s
are non-negative. Then
\[
\left\langle \xi,\eta\right\rangle =\sum_{i,j=1}^{d}x_{i}\left(\alpha_{i}^{\dagger}\middle|\eta\right)\geq0.
\]

Recalling that for perfectly
distinguishable pure states one has $\left(\alpha_{i}^{\dagger}\middle|\alpha_{j}\right)=\delta_{ij}$
\cite{TowardsThermo}, it is enough to take $\eta$ to be one of the
states $\left\{ \alpha_{i}\right\} _{i=1}^{d}$ to conclude that $x_{i}\geq0$
for every $i\in\left\{ 1,\ldots,d\right\} $, meaning that $\xi\in\mathsf{St}_{+}\left(\mathrm{A}\right)$.
\end{proof}
The self-dualising inner product, besides being a nice mathematical tool, has some operational meaning, because it provides a measure of the distinguishability of states, as explained in Appendix~\ref{subsec:dagger fidelity}. {Moreover, it is the starting point for extending the dagger to all transformations. This is done in Appendix~\ref{sec:dagger all}.}

\subsection{Existence of pure orthogonal projectors} \label{orthogonal projectors}

Now we show that we have orthogonal projectors on every face of the
state space. A consequence of diagonalisation is that all faces are
generated by perfectly distinguishable \emph{pure} states. Indeed,
every face $F$ is generated by a state $\omega$ in its relative
interior. $\omega$ can be diagonalised as $\omega=\sum_{i=1}^{r}p_{i}\alpha_{i}$,
where $r\leq d$, and $p_{i}>0$ for $i\in\left\{ 1,\ldots,r\right\} $.
By definition of face, this means that the states $\left\{ \alpha_{i}\right\} _{i=1}^{r}$
are in $F$, and therefore generate $F$. Consequently, there is an
effect $a$ that picks out the whole face
as the set of states $\rho$ such that $\left(a\middle|\rho\right)=1$. In the specific
case considered above, {it} is $a=\sum_{i=1}^{r}\alpha_{i}^{\dagger}$. Such faces are called \emph{exposed}.

Therefore the study of faces of sharp theories with purification reduces
to the study of normalised effects of the form $a_{\mathsf{I}}:=\sum_{i\in\mathsf{I}}\alpha_{i}^{\dagger}$, where  $\left\{ \alpha_{i}\right\} _{i=1}^{d}$ is a pure maximal set,
 $\mathsf{I}$ is a subset of $\left\{ 1,\ldots,d\right\} $ flagging the slits that are open in the experiment.
For every such $ a_{\mathsf{I}} $ we can define the two faces
\begin{enumerate}
\item $F_{\mathsf{I}}:=\left\{ \rho\in\mathsf{St}_{1}\left(\mathrm{A}\right):\left(a_{\mathsf{I}}\middle|\rho\right)=1\right\} $;
\item $F_{\mathsf{I}}^{\perp}:=\left\{ \rho\in\mathsf{St}_{1}\left(\mathrm{A}\right):\left(a_{\mathsf{I}}\middle|\rho\right)=0\right\} ${,}
\end{enumerate}
{in analogy with those of Definition~\ref{def:faces interference}.}
Clearly the effect ${a_{\mathsf{I}}^{\perp}:=}\sum_{i\notin\mathsf{I}}\alpha_{i}^{\dagger}$
defines the orthogonal face $F_{\mathsf{I}}^{\perp}$, as it occurs
with probability one on the states of $F_{\mathsf{I}}^{\perp}$. Note
that each of the effects $\left\{ \alpha_{i}^{\dagger}\right\} _{i\notin\mathsf{I}}$
occurs with zero probability on the states of $F_{\mathsf{I}}$.
\begin{Definition}
\label{def:An-orthogonal-projector}An \emph{orthogonal projector} (in the sense of \cite{Chiribella-informational})
on the face $F_{\mathsf{I}}$ is a transformation {$P_{\mathsf{I}}\in\mathsf{Transf}\left(\mathrm{A}\right)$}
such that
\begin{itemize}
\item if $\rho\in F_{\mathsf{I}}$, then $P_{\mathsf{I}}\rho=\rho$;
\item if $\rho\in F_{\mathsf{I}}^{\perp}$, then $P_{\mathsf{I}}\rho=0$.
\end{itemize}
\end{Definition}

We can prove the existence of a projector at least in one case, when
$\mathsf{I}=\left\{ 1,\ldots,d\right\} $. In this case $a_{\mathsf{I}}=u$,
so $F_{\mathsf{I}}=\mathsf{St}_{1}\left(\mathrm{A}\right)$, and $F_{\mathsf{I}}^{\perp}=\varnothing$.
Then it is enough to take $P_{\mathsf{I}}\doteq\mathcal{I}$. {However, sharp theories with purification admit projectors on \emph{every} face.}
\begin{Proposition}
\label{prop:existence projectors}Sharp theories with purification
have \emph{pure} projectors on every face $F_{\mathsf{I}}$. Furthermore
one has $uP_{\mathsf{I}}=a_{\mathsf{I}}$.
\end{Proposition}
\begin{proof}
Suppose $\rho$ is any state in $F_{\mathsf{I}}$, then $\left(a_{\mathsf{I}}\middle|\rho\right)=1$. By Proposition~\ref{prop:non-disturbing} we know that there is a
\emph{pure} transformation $P_{\mathsf{I}}$ such that $P_{\mathsf{I}}\rho=\rho$
for every $\rho\in F_{\mathsf{I}}$. We also have $\left(u\middle|P_{\mathsf{I}}\middle|\sigma\right)\leq\left(a_{\mathsf{I}}\middle|\sigma\right)$,
so if $\sigma\in F_{\mathsf{I}}^{\perp}$, we have $\left(u\middle|P_{\mathsf{I}}\middle|\sigma\right)=0$,
whence $P_{\mathsf{I}}\sigma=0$.

To prove that $uP_{\mathsf{I}}=a_{\mathsf{I}}$, first note that $\psi^{\dagger}P_{\mathsf{I}}=\psi^{\dagger}$
for every pure state $\psi\in F_{\mathsf{I}}$. Indeed $\psi^{\dagger}P_{\mathsf{I}}$
is pure by Purity Preservation, and we have $\left(\psi^{\dagger}\middle|P_{\mathsf{I}}\middle|\psi\right)=\left(\psi^{\dagger}\middle|\psi\right)=1$
because $P_{\mathsf{I}}\psi=\psi$ by definition. By Theorem~\ref{thm:duality},
we have $\psi^{\dagger}P_{\mathsf{I}}=\psi^{\dagger}$. Furthermore,
${\varphi}^{\dagger}P_{\mathsf{I}}=0$ for a pure state ${\varphi}\in F_{\mathsf{I}}^{\perp}$.
{Indeed, c}onsider
\[
\left({\varphi}^{\dagger}\middle|P_{\mathsf{I}}\middle|\chi\right)=\frac{1}{d}\sum_{{i}\in\mathsf{I}}\left({\varphi}^{\dagger}\middle|P_{\mathsf{I}}\middle|\alpha_{{i}}\right)+\frac{1}{d}\sum_{{i}\notin\mathsf{I}}\left({\varphi}^{\dagger}\middle|P_{\mathsf{I}}\middle|\alpha_{{i}}\right).
\]

The second term vanishes because $\alpha_{{i}}\in F_{\mathsf{I}}^{\perp}$
for ${i}\notin\mathsf{I}$. The first term vanishes because $P_{\mathsf{I}}\alpha_{{i}}=\alpha_{{i}}$
for ${i}\in\mathsf{I}$, and ${\varphi}$ is perfectly distinguishable from
any of the $\alpha_{{i}}$'s for ${i}\in\mathsf{I}$ by means of the observation-test
$\left\{ u-a_{\mathsf{I}},a_{\mathsf{I}}\right\} $, implying $\left({\varphi}^{\dagger}\middle|\alpha_{{i}}\right)=0$
\cite{TowardsThermo}. This means that ${\varphi}^{\dagger}P_{\mathsf{I}}$
occurs with zero probability on all states contained in $\chi$, and
since $\chi$ is complete \cite{Chiribella-purification}, ${\varphi}^{\dagger}P_{\mathsf{I}}=0$.
Now, when we calculate $uP_{\mathsf{I}}$, we separate the contribution
arising from states in orthogonal faces:
\[
uP_{\mathsf{I}}=\sum_{i\in\mathsf{I}}\alpha_{i}^{\dagger}P_{\mathsf{I}}+\sum_{i\notin\mathsf{I}}\alpha_{i}^{\dagger}P_{\mathsf{I}}=\sum_{i\in\mathsf{I}}\alpha_{i}^{\dagger}=a_{\mathsf{I}}
\]
This concludes the proof.
\end{proof}
In other words, $P_{\mathsf{I}}$ occurs with the same probability
as $a_{\mathsf{I}}$
{, thus satisfying one of the desiderata of Section~\ref{higher-order interference}}.  Moreover, extending some of the results in the proof of Proposition~\ref{prop:existence projectors}
by linearity, we obtain the dual statements of Definition~\ref{def:An-orthogonal-projector},
namely
\begin{itemize}
\item $\rho^{\dagger}P_{\mathsf{I}}=\rho^{\dagger}$ if $\rho\in F_{\mathsf{I}}$
\item $\rho^{\dagger}P_{\mathsf{I}}=0$ if $\rho\in F_{\mathsf{I}}^{\perp}$
\end{itemize}

Another consequence of Proposition~\ref{prop:existence projectors}
is that projectors actually project on their associated face, {viz.}\ for every normalised state $\rho$, $P_{\mathsf{I}}\rho=\lambda\sigma$,
where $\sigma$ is in $F_{\mathsf{I}}$, and $\lambda=\left(a_{\mathsf{I}}\middle|\rho\right)$.
Indeed, $\lambda=\left(u\middle|P_{\mathsf{I}}\middle|\rho\right)=\left(a_{\mathsf{I}}\middle|\rho\right)$.
If $\lambda\neq0$, which means $\rho\notin F_{\mathsf{I}}^{\perp}$,
then and $\left(a_{\mathsf{I}}\middle|\sigma\right)=\frac{1}{\lambda}\left(a_{\mathsf{I}}\middle|P_{\mathsf{I}}\middle|\rho\right)$.
However, we know that $a_{\mathsf{I}}P_{\mathsf{I}}=a_{\mathsf{I}}$,
so $\left(a_{\mathsf{I}}\middle|\sigma\right)=1$, showing that $\sigma\in F_{\mathsf{I}}$.

Furthermore, we can show that every projector $P_{\mathsf{I}}$ has
a \emph{complement} $P_{\mathsf{I}}^{\perp}$, which is the projector
associated with the effect $a_{\mathsf{I}}^{\perp}=\sum_{i\notin\mathsf{I}}\alpha_{i}^{\dagger}$,
which defines the orthogonal face $F_{\mathsf{I}}^{\perp}$. Clearly
$P_{\mathsf{I}}^{\perp}\rho=\left(a_{\mathsf{I}}^{\perp}\middle|\rho\right)\sigma$,
with $\sigma\in F_{\mathsf{I}}^{\perp}$. In particular, $P_{\mathsf{I}}^{\perp}\rho$
vanishes if and only if $\rho\in F_{\mathsf{I}}$.

{These properties are the starting point for proving the idempotence of projectors.}
\begin{Proposition}
\label{prop:properties projectors}Given a fixed pure maximal set
$\left\{ \alpha_{i}\right\} _{i=1}^{d}$ and {$\mathsf{I}\subseteq\left\{ 1,\ldots,d\right\} $},
one has $P_{\mathsf{I}}^{2}\doteq P_{\mathsf{I}}$. Moreover, if {$\mathsf{J}$ is another subset of $ \left\{ 1,\ldots,d\right\} $ disjoint from $ \mathsf{I} $,} then $P_{\mathsf{I}}P_{\mathsf{J}}\doteq0$.
\end{Proposition}
\begin{proof}
{Recall that for every state $\rho$, $P_{\mathsf{I}}\rho=\lambda\sigma$,
where $\sigma$ is in $F_{\mathsf{I}}$. Now, $P_{\mathsf{I}}$ leaves
$\sigma$ invariant by definition, so}
\[
P_{\mathsf{I}}^{2}\rho=\lambda P_{\mathsf{I}}\sigma=\lambda\sigma,
\]
so $P_{\mathsf{I}}^{2}\doteq P_{\mathsf{I}}$. To prove the other
property, note that if $\mathsf{I}$ and $\mathsf{J}$ are disjoint,
they define orthogonal faces. Indeed, suppose $\rho\in F_{\mathsf{I}}$,
then
\[
1=\mathrm{Tr}\:\rho=\left(a_{\mathsf{I}}\middle|\rho\right)+\left(a_{\mathsf{J}}\middle|\rho\right)+\sum_{i\notin\mathsf{I}\cup\mathsf{J}}\left(\alpha_{i}^{\dagger}\middle|\rho\right),
\]
which implies $\left(a_{\mathsf{J}}\middle|\rho\right)=0$ because $\left(a_{\mathsf{I}}\middle|\rho\right)=1$.
Hence $\rho\in F_{\mathsf{J}}^{\perp}$. Now, given \emph{any} normalised
state $\rho$, $P_{\mathsf{I}}P_{\mathsf{J}}\rho=0$ because $P_{\mathsf{J}}\rho$
is proportional to a state in $F_{\mathsf{I}}^{\perp}$. This proves
that $P_{\mathsf{I}}P_{\mathsf{J}}\doteq0$.
\end{proof}
This result shows that, once a pure maximal set $\left\{ \alpha_{i}\right\}_{{i=1}}^{{d}} $
is fixed, whenever we have a partition $\left\{ \mathsf{I}_{j}\right\} $
of $\left\{ 1,\ldots,d\right\} $, the test $\left\{ P_{\mathsf{I}_{j}}\right\} $
is a von Neumann measurement. The only thing left to check is that
$\sum_{j}uP_{\mathsf{I}_{j}}=u$, which is a sufficient condition
for a set of transformations to be a test in sharp theories with purification. This is satisfied because, recalling
Proposition~\ref{prop:existence projectors},
\[
\sum_{j}uP_{\mathsf{I}_{j}}=\sum_{j}a_{\mathsf{I}_{j}}=\sum_{i=1}^{d}\alpha_{i}^{\dagger}=u.
\]

Because of the properties proved above, von Neumann measurements are
repeatable and minimally disturbing measurements in the sense of Refs.~\cite{Chiribella-Yuan2014,Chiribella-Yuan2015}.
Indeed, $a_{\mathsf{I}_{j}}P_{\mathsf{I}_{j}}=a_{\mathsf{I}_{j}}$,
and
\[
a_{\mathsf{I}_{j}}\sum_{k}P_{\mathsf{I}_{k}}=a_{\mathsf{I}_{j}}P_{\mathsf{I}_{j}}+\sum_{k\neq j}a_{\mathsf{I}_{j}}P_{\mathsf{I}_{k}}=a_{\mathsf{I}_{j}},
\]
because for $k\neq j$ the $P_{\mathsf{I}_{k}}$'s project on faces
orthogonal to $F_{\mathsf{I}_{j}}$.

{The next proposition concerns the interplay between orthogonal projectors and the dagger.}
\begin{Proposition}\label{prop:self-adjoint}
For every normalised state $\rho$, and for every projector $P_{\mathsf{I}}$
on a face $F_{\mathsf{I}}$, one has $\left(P_{\mathsf{I}}\rho\right)^{\dagger}=\rho^{\dagger}P_{\mathsf{I}}$.
\end{Proposition}
\begin{proof}
First of all, note that $0\leq\left\Vert P_{\mathsf{I}}\rho\right\Vert \leq1$,
and it vanishes if and only if $\rho\in F_{\mathsf{I}}^{\perp}$.
If $\rho\in F_{\mathsf{I}}^{\perp}$, then $\rho^{\dagger}P_{\mathsf{I}}=0$,
so the statement is trivially true. Now suppose $\left\Vert P_{\mathsf{I}}\rho\right\Vert >0$.
We will first prove the statement for normalised pure states $\psi$, then
it is sufficient to extend it by linearity to all states. We
will make use of the uniqueness of the dagger for normalised pure
states. Then the statement is equivalent to proving
\[
\left(\frac{P_{\mathsf{I}}\psi}{\left\Vert P_{\mathsf{I}}\psi\right\Vert }\right)^{\dagger}=\frac{\psi^{\dagger}P_{\mathsf{I}}}{\left\Vert P_{\mathsf{I}}\psi\right\Vert },
\]
Noting that the term in brackets is a \emph{normalised} pure state
(by Purity Preservation), and that the RHS is a pure effect (again
by Purity Preservation), by the uniqueness of the dagger for normalised
pure states (cf.\ Theorem~\ref{thm:duality}), it is enough to prove
that
\[
\frac{\left(\psi^{\dagger}P_{\mathsf{I}}\middle|P_{\mathsf{I}}\psi\right)}{\left\Vert P_{\mathsf{I}}\psi\right\Vert ^{2}}=1;
\]
in other words that $\left(\psi^{\dagger}P_{\mathsf{I}}\middle|P_{\mathsf{I}}\psi\right)=\left\Vert P_{\mathsf{I}}\psi\right\Vert ^{2}$.
Recall that $P_{\mathsf{I}}^{2}\doteq P_{\mathsf{I}}$ (Proposition~\ref{prop:properties projectors}),
so $\left(\psi^{\dagger}P_{\mathsf{I}}\middle|P_{\mathsf{I}}\psi\right)=\left(\psi^{\dagger}\middle|P_{\mathsf{I}}\middle|\psi\right)$.
Now, $P_{\mathsf{I}}\psi=\left\Vert P_{\mathsf{I}}\psi\right\Vert \psi'$,
where $\psi'$ is a pure state in $F_{\mathsf{I}}$. We have $\left(\psi^{\dagger}P_{\mathsf{I}}\middle|P_{\mathsf{I}}\psi\right)=\left\Vert P_{\mathsf{I}}\psi\right\Vert \left(\psi^{\dagger}\middle|\psi'\right)$.
We only need to prove that $\left(\psi^{\dagger}\middle|\psi'\right)=\left\Vert P_{\mathsf{I}}\psi\right\Vert $.
Recall that $\left(\psi^{\dagger}\middle|\psi'\right)=\left(\psi^{'\dagger}\middle|\psi\right)$
by Lemma~\ref{lem:inner product}, and that $\psi^{'\dagger}P_{\mathsf{I}}=\psi^{'\dagger}$
as $\psi'\in F_{\mathsf{I}}$, thus
\[
\left(\psi^{\dagger}\middle|\psi'\right)=\left(\psi^{'\dagger}\middle|P_{\mathsf{I}}\middle|\psi\right)=\left\Vert P_{\mathsf{I}}\psi\right\Vert \left(\psi^{'\dagger}\middle|\psi'\right)=\left\Vert P_{\mathsf{I}}\psi\right\Vert .
\]
By the uniqueness of the dagger for normalised pure states we conclude
that $\left(\frac{P_{\mathsf{I}}\psi}{\left\Vert P_{\mathsf{I}}\psi\right\Vert }\right)^{\dagger}=\frac{\psi^{\dagger}P_{\mathsf{I}}}{\left\Vert P_{\mathsf{I}}\psi\right\Vert }$,
namely $\left(P_{\mathsf{I}}\psi\right)^{\dagger}=\psi^{\dagger}P_{\mathsf{I}}$.
\end{proof}
A consequence of this proposition is that orthogonal projectors
play nicely with the inner product of Lemma~\ref{lem:inner product},
namely for every $\xi,\eta\in\mathsf{St}_{\mathbb{R}}\left(\mathrm{A}\right)$
one has
\begin{equation} \label{self-dual projector}
\left\langle P_{\mathsf{I}}\xi,\eta\right\rangle =\left\langle \xi,P_{\mathsf{I}}\eta\right\rangle.\end{equation}  In other words, projections are symmetric with respect to the inner product.

The last property we need is a generalisation of the results of Proposition~\ref{prop:properties projectors}.
\begin{Proposition}\label{prop:product projectors}
Fixing a pure maximal set $\left\{ \alpha_{i}\right\} _{i=1}^{d}$,
and considering $\mathsf{I},\mathsf{J}\subseteq\left\{ 1,\ldots,d\right\} $,
we have $P_{\mathsf{I}}P_{\mathsf{J}}\doteq P_{\mathsf{I}\cap\mathsf{J}}$.
\end{Proposition}
\begin{proof}
First let us prove that
\begin{equation}
P_{\mathsf{I}}P_{\mathsf{J}}\rho=\left\Vert P_{\mathsf{I}}P_{\mathsf{J}}\rho\right\Vert \rho'\label{eq:PIPJ}
\end{equation}
for every normalised state $\rho$, where $\rho'\in F_{\mathsf{I}\cap\mathsf{J}}$.
Let us show that $\left\Vert P_{\mathsf{I}}P_{\mathsf{J}}\rho\right\Vert =\left(a_{\mathsf{I}\cap\mathsf{J}}\middle|\rho\right)$.
By Proposition~\ref{prop:existence projectors}, $\left(u\middle|P_{\mathsf{I}}P_{\mathsf{J}}\middle|\rho\right)=\left(a_{\mathsf{I}}\middle|P_{\mathsf{J}}\middle|\rho\right)$.
Now, recalling that $a_{\mathsf{I}}=\sum_{i\in\mathsf{I}}\alpha_{i}^{\dagger}$,
\[
\left(a_{\mathsf{I}}\middle|P_{\mathsf{J}}\middle|\rho\right)=\sum_{i\in\mathsf{I}\cap\mathsf{J}}\left(\alpha_{i}^{\dagger}\middle|P_{\mathsf{J}}\middle|\rho\right)+\sum_{i\in\mathsf{I}\backslash\mathsf{J}}\left(\alpha_{i}^{\dagger}\middle|P_{\mathsf{J}}\middle|\rho\right)=\sum_{i\in\mathsf{I}\cap\mathsf{J}}\left(\alpha_{i}^{\dagger}\middle|\rho\right)=\left(a_{\mathsf{I}\cap\mathsf{J}}\middle|\rho\right),
\]
where we have used the fact that $\alpha_{i}^{\dagger}P_{\mathsf{J}}=\alpha_{i}^{\dagger}$
if $i\in\mathsf{J}$, and $\alpha_{i}^{\dagger}P_{\mathsf{J}}=0$
if $i\notin\mathsf{J}$. If $\rho\in F_{\mathsf{I}\cap\mathsf{J}}^{\perp}$,
both the LHS and the RHS of Equation~\eqref{eq:PIPJ} vanish, and the
statement is trivially satisfied. Now, let us assume $\rho\notin F_{\mathsf{I}\cap\mathsf{J}}^{\perp}$,
in this case $\left(a_{\mathsf{I}\cap\mathsf{J}}\middle|\rho\right)>0$.
We wish to prove that $\left(a_{\mathsf{I}\cap\mathsf{J}}\middle|P_{\mathsf{I}}P_{\mathsf{J}}\middle|\rho\right)=\left(a_{\mathsf{I}\cap\mathsf{J}}\middle|\rho\right)$.
Recalling the expression of $a_{\mathsf{I}\cap\mathsf{J}}$, we have
\[
\sum_{i\in\mathsf{I}\cap\mathsf{J}}\left(\alpha_{i}^{\dagger}\middle|P_{\mathsf{I}}P_{\mathsf{J}}\middle|\rho\right)=\sum_{i\in\mathsf{I}\cap\mathsf{J}}\left(\alpha_{i}^{\dagger}\middle|P_{\mathsf{J}}\middle|\rho\right)=\sum_{i\in\mathsf{I}\cap\mathsf{J}}\left(\alpha_{i}^{\dagger}\middle|\rho\right)=\left(a_{\mathsf{I}\cap\mathsf{J}}\middle|\rho\right),
\]
again by the properties of $P_{\mathsf{I}}$ and $P_{\mathsf{J}}$.
This means that $P_{\mathsf{I}}P_{\mathsf{J}}$ maps every normalised
state to a state of $F_{\mathsf{I}\cap\mathsf{J}}$, up to normalisation.

Now let us prove that $\left(P_{\mathsf{I}}P_{\mathsf{J}}\right)^{2}\doteq P_{\mathsf{I}}P_{\mathsf{J}}$.
First note that $F_{\mathsf{I}\cap\mathsf{J}}\subseteq F_{\mathsf{I}}$.
Indeed, suppose $\rho\in F_{\mathsf{I}\cap\mathsf{J}}$, then
\[
\left(a_{\mathsf{I}}\middle|\rho\right)=\sum_{i\in\mathsf{I}\cap\mathsf{J}}\left(\alpha_{i}^{\dagger}\middle|\rho\right)+\sum_{i\in\mathsf{I}\backslash\mathsf{J}}\left(\alpha_{i}^{\dagger}\middle|\rho\right)=\left(a_{\mathsf{I}\cap\mathsf{J}}\middle|\rho\right)=1,
\]
where we have used the fact that $\left(\alpha_{i}^{\dagger}\middle|\rho\right)=0$
if $i\notin\mathsf{I}\cap\mathsf{J}$. By a similar argument, $F_{\mathsf{I}\cap\mathsf{J}}\subseteq F_{\mathsf{J}}$.
Now, $P_{\mathsf{I}}P_{\mathsf{J}}\rho=\left\Vert P_{\mathsf{I}}P_{\mathsf{J}}\rho\right\Vert \rho'$,
with $\rho'\in F_{\mathsf{I}\cap\mathsf{J}}$. Then $\left(P_{\mathsf{I}}P_{\mathsf{J}}\right)^{2}\rho=\left\Vert P_{\mathsf{I}}P_{\mathsf{J}}\rho\right\Vert P_{\mathsf{I}}P_{\mathsf{J}}\rho'$.
However, $\rho'\in F_{\mathsf{J}}$, so $P_{\mathsf{J}}\rho'=\rho'$,
and, similarly, $\rho'\in F_{\mathsf{I}}$, so $P_{\mathsf{I}}\rho'=\rho'$.
Consequently,
\[
\left(P_{\mathsf{I}}P_{\mathsf{J}}\right)^{2}\rho=\left\Vert P_{\mathsf{I}}P_{\mathsf{J}}\rho\right\Vert \rho'=P_{\mathsf{I}}P_{\mathsf{J}}\rho,
\]
proving that $\left(P_{\mathsf{I}}P_{\mathsf{J}}\right)^{2}\doteq P_{\mathsf{I}}P_{\mathsf{J}}$.

Now let us prove that for every $\xi\in\mathsf{St}_{\mathbb{R}}\left(\mathrm{A}\right)$,
we have $\left(P_{\mathsf{I}}P_{\mathsf{J}}\xi\right)^{\dagger}=\xi^{\dagger}P_{\mathsf{I}}P_{\mathsf{J}}$.
Following the lines of proof of Proposition~\ref{prop:self-adjoint},
let us show that this is true when $\xi$ is a normalised pure state
$\psi$. This boils down to showing that
\[
\left(\psi^{\dagger}P_{\mathsf{I}}P_{\mathsf{J}}\middle|P_{\mathsf{I}}P_{\mathsf{J}}\psi\right)=\left\Vert P_{\mathsf{I}}P_{\mathsf{J}}\psi\right\Vert ^{2}.
\]

The proof goes on as for Proposition~\ref{prop:self-adjoint}, noting
that if $\psi'\in F_{\mathsf{I}\cap\mathsf{J}}$, then $\psi'^{\dagger}P_{\mathsf{I}}P_{\mathsf{J}}=\psi'^{\dagger}$
because $\psi'^{\dagger}P_{\mathsf{I}}=\psi'^{\dagger}$ as $\psi'\in F_{\mathsf{I}}$,
and, similarly, $\psi'^{\dagger}P_{\mathsf{J}}=\psi'^{\dagger}$ as
$\psi'\in F_{\mathsf{J}}$. Eventually we find that for pure states
$\left(P_{\mathsf{I}}P_{\mathsf{J}}\psi\right)^{\dagger}=\psi^{\dagger}P_{\mathsf{I}}P_{\mathsf{J}}$,
and by linearity this means that $\left(P_{\mathsf{I}}P_{\mathsf{J}}\xi\right)^{\dagger}=\xi^{\dagger}P_{\mathsf{I}}P_{\mathsf{J}}$.

A consequence of this property is that $\left\langle P_{\mathsf{I}}P_{\mathsf{J}}\xi,\eta\right\rangle =\left\langle \xi,P_{\mathsf{I}}P_{\mathsf{J}}\eta\right\rangle $,
for all $\xi,\eta\in\mathsf{St}_{\mathbb{R}}\left(\mathrm{A}\right)$.
These linear maps on $\mathsf{St}_{\mathbb{R}}\left(\mathrm{A}\right)$
are such that $\mathsf{St}_{\mathbb{R}}\left(\mathrm{A}\right)=\mathrm{im}\:P_{\mathsf{I}}P_{\mathsf{J}}\oplus\ker P_{\mathsf{I}}P_{\mathsf{J}}$,
and $\ker P_{\mathsf{I}}P_{\mathsf{J}}$ is the orthogonal subspace
to $\mathrm{im}\:P_{\mathsf{I}}P_{\mathsf{J}}$, hence it is uniquely
defined once $\mathrm{im}\:P_{\mathsf{I}}P_{\mathsf{J}}$ is fixed.
Note that for any projector $P_{\mathsf{I}}$ we have $\mathrm{im}\:P_{\mathsf{I}}=\mathrm{span}\:F_{\mathsf{I}}$,
and we have just proved that $\mathrm{im}\:P_{\mathsf{I}}P_{\mathsf{J}}=\mathrm{span}\:F_{\mathsf{I}\cap\mathsf{J}}=\mathrm{im}\:P_{\mathsf{I}\cap\mathsf{J}}$.
Having the same image, and consequently the same kernel, $P_{\mathsf{I}}P_{\mathsf{J}}$
and $P_{\mathsf{I}\cap\mathsf{J}}$ agree on a basis of $\mathsf{St}_{\mathbb{R}}\left(\mathrm{A}\right)$,
therefore they agree also on all states of $\mathrm{A}$, meaning
that $P_{\mathsf{I}}P_{\mathsf{J}}\doteq P_{\mathsf{I}\cap\mathsf{J}}$.
\end{proof}

\subsection{Main result}

Proposition $29$ of \cite{Barnum-interference} asserts that theories satisfying two postulates, Strong Symmetry and Projectivity, have higher-order interference if and only if their projectors (in our terminology here) preserve purity.  A close examination of its proof, and those of all lemmas and propositions used in its proof---notably {L}emma 22 and {P}ropositions 18, 25, 26, and 28 of \cite{Barnum-interference}---reveals that only premises weaker than the conjunction of Strong Symmetry and Projectivity are used:  self-duality, the ``spectral-like decomposition'' of effects as in Lemma~\ref{lem:normalised effects} above, the fact that faces are determined by subsets of maximal distinguishable sets of states as in Section~\ref{orthogonal projectors} above, the existence of projectors onto each face in the sense of Definition~\ref{def:An-orthogonal-projector} above, and the fact that these are symmetric with respect to the self-dualising inner product (i.e.,\ orthogonal projectors), and satisfy Proposition~\ref{prop:product projectors} above. We have established these weaker premises for sharp theories with purification, and moreover, we have established in {P}roposition~\ref{prop:existence projectors} that their projectors preserve purity, so we have proved:

\begin{thm}\label{theorem: main}
In any sharp theory with purification there can be no $n$th order interference for $n \geq 3$.
\end{thm}

\subsection{Jordan-algebraic structure}
Our results also imply that systems, and therefore also the ``subsystems'' associated with their faces, are operationally equivalent to finite-dimensional Jordan-algebraic systems.  These are systems $\rA$ for which
$\mathsf{St}_{+}\left(\rA\right)$ is the cone of squares in a finite-dimensional Euclidean Jordan algebra (EJA) and
$\mathsf{Eff}_{+}\left(\rA\right)$ is identified with the same cone, with evaluation of effects on states given by the inner product and
the Jordan unit as the deterministic effect.  (See
\cite{BarnumGraydonWilceCCEJA} for more on Jordan algebraic operational systems, and \cite{Alfsen-Shultz} for a mathematical treatment.)

\begin{thm}\label{theorem:Jordan}
In a sharp theory with purification, every system $A$ has both $\mathsf{St}_{+}\left({\rA}\right)$ and
$\mathsf{Eff}_{+}\left({\rA}\right)$ isomorphic to the cone of squares in a Euclidean Jordan algebra (EJA) via isomorphisms
$S$ and $T$ such that $\left(a\middle|\rho\right) = \left\langle Ta , S\rho \right\rangle$, where $\left\langle \bullet , \bullet \right\rangle$ is the canonical inner product on the EJA, and  $T$ takes the deterministic effect to the Jordan unit.
\end{thm}
\begin{proof}
The proof uses results of Alfsen and Shultz
\cite{Alfsen-Shultz78}, for which we refer
to \cite{Alfsen-Shultz}.  Theorem 9.33 in~\cite{Alfsen-Shultz} implies that finite-dimensional systems with symmetry of transition probabilities (STP), a type of projection operator they call  ``compression'' associated with every face, and whose compressions preserve purity, have state spaces affinely isomorphic to the state spaces of Euclidean Jordan algebras.  Sharp theories with purification satisfy STP, as noted following {L}emma~\ref{lem:inner product} above. Our projectors are easily shown to be examples of compressions by the same argument as in {T}heorem 17 of \cite{Barnum-interference};  this argument uses only properties satisfied by our projectors (the same ones needed in the proof of {T}heorem
\ref{theorem: main}, except for {P}urity {P}reservation) and does not need {S}trong {S}ymmetry.  As shown above, our projectors also preserve purity.
\end{proof}

Since faces of Jordan-algebraic systems are also Jordan-algebraic
({to see this, combine a result of Iochum \cite{Iochum} ({T}heorem 5.32 in \cite{Alfsen-Shultz}), whose finite dimensional case is that all faces of EJAs are the positive part of the images of compressions, with the facts (cf.\ pp.\ 22--26 of \cite{Alfsen-Shultz}) that every face of the cone of squares is the image of such a compression $P$ (\cite{Alfsen-Shultz}, {L}emma 1.39), and also a Jordan subalgebra whose unit is the image of the order unit under $P$ (\cite{Alfsen-Shultz}, {P}roposition 1.43).}),
so are the faces of state spaces in sharp theories with purification.  However, it is not the case that in sharp theories with purification, each face of a system is necessarily isomorphic to a stand-alone system of the theory (an \emph{object} of the category, in the categorical formulation), { but, it is always possible to extend the theory such that they are.} Every category has a \emph{Cauchy completion}: this is a minimal extension of the category such that every idempotent morphism $\pi: \rA \rightarrow \rA$ can be written as a retraction-section pair, i.e.,\ as the composition $\pi = \sigma \circ \rho$,
with $\rho: \rA \rightarrow \rB$ and
$\sigma: \rB \rightarrow\rA$, such that the reverse composition $\rho \circ \sigma$ is the identity morphism on $B$.  When the idempotents are projectors $P$ like the ones we consider here, $B$ will be a system isomorphic to the face $\mathrm{im}_+(P)$.  Of course, since there may be idempotents beyond the projectors onto faces (for example, decoherence of a set
of orthogonal subspaces, or damping to a fixed state, in quantum theory), Cauchy completion of an operational theory $T$ may add many objects in addition to ones isomorphic to faces of systems of $T$; indeed, for many operational theories (e.g.,\ ones possessing idempotent decoherence maps) this will add some classical systems.
{{This is indeed the case for quantum theory where the Cauchy completion leads to the category of finite-dimensional C*-algebras and completely positive maps  \cite{coecke2017two}.}}  {The Cauchy completion can be thought of as adding in all operationally accessible systems that
 can be simulated on the physical system via a consistent restriction on the allowed states, effects and transformations.} The Cauchy completion of a sharp theory with purification will likely satisfy the Ideal Compression postulate by virtue of containing the faces that are images of orthogonal projectors; but there are also non-Cauchy complete theories that satisfy it, e.g.,\ the category CPM of finite-dimensional quantum systems and CP maps, in which all systems, and also all images of orthogonal projectors as defined above, are fully coherent quantum systems, but there are no classical systems.

In \cite{BarnumGraydonWilceCCEJA}, some categories, including dagger-compact-closed categories, of Jordan algebraic systems were constructed; these categories are equivalent to operational theories as we use the term here.  Although sharp theories with purification also have Jordan algebraic state and effect spaces, it is interesting to note that some of the explicit examples in
\cite{TowardsThermo,Purity} involve composites different from those that would be obtained in the categories considered in \cite{BarnumGraydonWilceCCEJA} for systems with the same state spaces.  On the other hand, the category combining real and quaternionic systems in \cite{BarnumGraydonWilceCCEJA} does \emph{not} satisfy Purity Preservation by parallel composition and hence falls outside the class of sharp theories with purification, although its filters do preserve purity.  Of course, the failure of Purity Preservation by parallel composition seems likely to allow phenomena like the nonextensiveness of entropy when products of states are taken, which could warrant focusing on sharp theories with purification in thermodynamically motivated work such as \cite{TowardsThermo}.

That Jordan-algebraic systems lack higher-order interference was shown by Barnum and Ududec (\cite{UBEinterferenceUnpublished}; announced in \cite{BarnumReconstructionTalk}) and by Niestegge \cite{NiesteggeJordanInterference}; combining this with {T}heorem \ref{theorem:Jordan} gives another way to see that our results on sharp theories with purification imply the absence of higher-order interference. {Moreover, as not all EJAs satisfy our postulates, it is clear that our postulates are sufficient but not necessary conditions for ruling out higher-order interfence.}

\section{Discussion and conclusions} \label{end}

{We proved that in sharp theories with purification multi-slit experiments must have a pure projector structure and, moreover, such theories exhibit at most second-order interference. Hence these theories are, at least conceptually, very ``close'' to quantum theory. Moreover, recent work has shown that sharp theories with purification are close to quantum theory in terms of other physical and information processing features. Indeed, such theories possess quantum-like contextuality behaviour \cite{Chiribella-Yuan2014,Chiribella-Yuan2015}, quantum-like computation \cite{Lee-Selby-Grover,Control-reversible}, and quantum-like thermodynamic properties \cite{Diagonalization,TowardsThermo,Purity}. Recall from Section~\ref{sharp theories with purification}} that quantum theory is not the only example of a generalised probabilistic theory satisfying these principles. Hence {C}ausality, {P}urity {P}reservation, {P}ure {S}harpness, and {P}urification do not recover the entire quantum formalism.

{However, if one were to introduce the Ideal Compression and Local Discriminability principles of the reconstruction of quantum theory due to Chiribella, D`Ariano, and Perinotti \cite{Chiribella-informational}, one would indeed regain the entire quantum formalism. Indeed, both additional principles are necessary: Local Discriminability to preclude real quantum theory and Ideal Compression to preclude the contrived---yet admissible---example of the theory in which all systems are composites of qubits.  Sharp theories with purification thus serve as a fertile test-bed for physics that is conceptually quite close to that predicted by the quantum world, but which may diverge from it in certain small,} yet interesting, ways.

\subsection{Finding higher-order interference}

To date there has been no experiment that has found higher-order interference,
at least, none that cannot be explained by taking into account the fact that the ``sets of histories are not mutually exclusive'' \cite{Sorkin1, sinha2015superposition}. However, this might be due to the specific experimental set-up employed, rather than a fundamental preclusion of higher-order interference in nature. We show here that many of the properties needed to rule out observing higher-order interference are in fact quite natural assumptions which appear to be suggested by the experimental set-up employed. This suggests that the experimental set-up itself may implicitly rule out observing higher-order interference from the outset.

The main result of the current work is that sharp theories with purification can never exhibit higher-order interference in any experiment. However, in a wider class of theories, we still will not observe higher-order interference in a particular experiment if the following three conditions are met{;} hence, to have any chance of observing higher-order interference{,} experiments must be designed in order to try to violate these conditions.
\begin{enumerate}
\item The transformations corresponding to blocking slits   satisfy: $T_{\mathsf{I}}T_{\mathsf{J}}=T_{\mathsf{I}\cap \mathsf{J}}$. By this we mean that they share several properties with the projectors $P_{\mathsf{I}}$ of Section~\ref{proofs}: if we define the effects $a_{\mathsf{I}} = {uT_{\mathsf{I}}}$ and the faces $F_{\mathsf{I}}$ and $F_{\mathsf{I}}^\perp$ as in {S}ection~\ref{orthogonal projectors}, i.e.,\ as the $1$-set and $0$-set of $a_{\mathsf{I}}$, then the $T_{\mathsf{I}}$ are assumed to be
  	\emph{orthogonal projectors} in the sense of {D}efinition~\ref{def:An-orthogonal-projector}, and to be both idempotent and ``orthogonal'' ($T_{\mathsf{I}}T_{\mathsf{J}} = 0$) if $\mathsf{I}$ and $\mathsf{J}$ are disjoint (as in Proposition~\ref{prop:properties projectors}).
\item The $T_{\mathsf{I}}${'s} map pure states to pure states
\item The $T_{\mathsf{I}}${'s} are self-adjoint.
\end{enumerate}

 The first of these is generally expected as only those slits belonging to both $\mathsf{I}$ and $\mathsf{J}$ will not be blocked by either $T_{\mathsf{I}}$ or $T_{\mathsf{J}}${,} and so should hold in this experimental set-up for any theory that can describe it.

The second assumption, which is also natural given the multi-slit set-up, is that, in an idealised scenario, the slits should not introduce fundamental noise. That is, if an input state $\rho$ is pure, i.e.,\ has no classical noise associated with it, then $T_{\mathsf{I}}\rho$ should also be pure. Hence it appears natural to assume that $T_{\mathsf{I}}$ maps pure states to pure states. Violating this principle by just adding noise to the experiment does not seem likely to demonstrate higher-order interference{. A} more plausible way to violate this however would be if the particle passing through the slits were to become entangled with some degree of freedom associated with them, if we do not have access to this degree of freedom then this would send a pure input to a mixed state.

The final assumption is far less general than the others, as it places a constraint on the theory. That is, to even discuss whether a transformation is self-adjoint {(cf.\ also Appendix~\ref{sec:dagger all})}, one requires that the theory itself be self-dual. To fully understand what this assumption entails{,} one needs an operational or physical interpretation of the self-dualising inner product ({see \cite{Selby-dagger} for an example of such an interpretation}). However, intuitively this notion reflects the inherent symmetry of the experimental set{-}up. Here one could consider propagation from the source to the effect or from the effect to the source as being ``dual'' to one another and, moreover, that the physical blocking of slits has an equivalent effect in either situation. That is, the assumption of self-adjointness corresponds to the statement that the projector has an equivalent action on the effects associated with a particular slit as it does on the states which can pass through them.

If an experiment satisfies these assumptions then for any self-dual theory
it was shown in  \cite{Barnum-interference} (Proposition 29) that we will not see higher-order interference in this experiment. Hence any set of physical principles which ensure these assumptions hold will rule out higher-order interference. Because the mathematical assumptions involved in formalising a multi-slit experiment are so natural when interpreted operationally, perhaps one should search for higher-order interference in set-ups that don't seem to preclude it from the outset. This could involve ``asymmetric'' multi-slit set-ups that are {not} obviously time-symmetric in an arbitrary generalised probabilistic theory. One could also consider experiments that search for higher-order phases \cite{Control-reversible}, a reformulation of higher-order interference that makes no reference to projectors and hence does not preclude certain generalised theories from the outset. The assumption that nature is self-dual could also be rejected{;} this poses the question as to whether it is possible to find {a} direct experimental test of this principle.

\vspace{6pt}


\section*{Acknowledgments}{The authors thank J.\ van de Wetering for pointing out an error in Lemma~\ref{lem:normalised effects} in the previous version of the paper. The authors also thank J.\ Barrett,  and J.\ J.\ Barry for encouragement while writing the current paper. This work was supported by EPSRC grants through the Controlled Quantum Dynamics Centre for Doctoral Training, the UCL Doctoral Prize Fellowship, and an Oxford doctoral training scholarship, and also by Oxford-Google DeepMind Graduate Scholarship. We also acknowledge financial support from the European Research Council (ERC Grant Agreement no 337603), the Danish Council for Independent Research (Sapere Aude) and VILLUM FONDEN via the QMATH Centre of Excellence (Grant No. 10059).  This work
began while the authors were attending the ``Formulating
and Finding Higher-order Interference'' workshop at
the Perimeter Institute. Research at Perimeter Institute
is supported by the Government of Canada through the
Department of Innovation, Science and Economic Development
Canada and by the Province of Ontario through
the Ministry of Research, Innovation and Science.}


\appendix

\setcounter{equation}{0}
\renewcommand{\theequation}{A\arabic{equation}}

\section{Norms and fidelity}
\vspace{-6pt}
\subsection{Operational norm and dagger norm\label{subsec:operational vs dagger}}

In Ref.~\cite{Chiribella-purification} the operational norm for
every vector $\xi\in\mathsf{St}_{\mathbb{R}}\left(\mathrm{A}\right)$
was introduced:
\[
\left\Vert \xi\right\Vert :=\sup_{a\in\mathsf{Eff}\left(\mathrm{A}\right)}\left(a\middle|\xi\right)-\inf_{a\in\mathsf{Eff}\left(\mathrm{A}\right)}\left(a\middle|\xi\right)
\]

As pointed out in \cite{Chiribella-purification}, {in quantum theory} the operational
norm coincides with the trace norm. The analogy
is apparent also in sharp theories with purification.
\begin{Proposition}
Let $\xi\in\mathsf{St}_{\mathbb{R}}\left(\mathrm{A}\right)$ be diagonalised
as $\xi=\sum_{i=1}^{d}x_{i}\alpha_{i}$. Then $\left\Vert \xi\right\Vert =\sum_{i=1}^{d}\left|x_{i}\right|$.
\end{Proposition}
\begin{proof}
{Let us separate the terms with non-negative eigenvalues from the terms
with negative eigenvalues, so that we can write $\xi=\xi_{+}-\xi_{-}$,
where $\xi_{+}:=\sum_{x_{i}\geq0}x_{i}\alpha_{i}$, and $\xi_{-}=\sum_{x_{i}<0}\left(-x_{i}\right)\alpha_{i}$.
Clearly, $\xi_{+},\xi_{-}\in\mathsf{St}_{+}\left(\mathrm{A}\right)$.
In order to achieve the supremum of $\left(a\middle|\xi\right)$ we must
have $\left(a\middle|\xi_{-}\right)=0$. Moreover},
\[
\left(a\middle|\xi_{+}\right)=\sum_{x_{i}\geq0}x_{i}\left(a\middle|\alpha_{i}\right)\leq\sum_{x_{i}\geq0}x_{i}
\]
since $\left(a\middle|\alpha_{i}\right)\leq1$ for every $i$. The supremum
of $\left(a\middle|\xi_{+}\right)$ is achieved by $a=\sum_{x_{i}\geq0}\alpha_{i}^{\dagger}$.
Hence $\sup_{a}\left(a\middle|\xi\right)=\sum_{x_{i}\geq0}x_{i}$. By a similar
argument, one shows that $\inf_{a}\left(a\middle|\xi\right)=\sum_{x_{i}<0}x_{i}$.
Therefore
\[
\left\Vert \xi\right\Vert =\sum_{x_{i}\geq0}x_{i}+\sum_{x_{i}<0}\left(-x_{i}\right)=\sum_{i=1}^{d}\left|x_{i}\right|.
\]
\end{proof}
For $p\geq1$, the $p$-norm of a vector $\mathbf{x}\in\mathbb{R}^{d}$
	is defined as $\left\Vert \mathbf{x}\right\Vert _{p}:=\left(\sum_{i=1}^{d}\left|x_{i}\right|^{p}\right)^{\frac{1}{p}}$, thus we have $\left\Vert \xi\right\Vert =\left\Vert \mathbf{x}\right\Vert _{1}$, {where $\mathbf{x}$ is} the spectrum of $\xi$.

In sharp theories with purification we have an additional norm, the
dagger norm, defined in Section~\ref{subsec:Self-duality}. The dagger norm of a vector ${\xi\in\mathsf{St}_{\mathbb{R}}\left( \rA\right)}$ is $\left\Vert {\xi}\right\Vert _{\dagger}=\sqrt{\sum_{i=1}^{d}{x}_{i}^{2}}$,
where the ${x}_{i}$'s are the eigenvalues of ${\xi}$. {It is obvious} from the very definition that $\left\Vert \xi\right\Vert _{\dagger}=\left\Vert \mathbf{x}\right\Vert _{2}$. Thanks to these results following from diagonalisation, we can derive
the standard bound{s} between the two norms, by making use of the well-known
bounds $\left\Vert \mathbf{x}\right\Vert _{2}\leq\left\Vert \mathbf{x}\right\Vert _{1}\leq\sqrt{d}\left\Vert \mathbf{x}\right\Vert _{2}$,
which impl{y}
\begin{equation}\label{eq:bounds norms}
\left\Vert \xi\right\Vert _{\dagger}\leq\left\Vert \xi\right\Vert \leq\sqrt{d}\left\Vert \xi\right\Vert _{\dagger}.
\end{equation}
Note that, unlike Ref.~\cite{Muller-blackhole}, here the bound{s are}
derived without assuming Bit Symmetry \cite{Muller-self-duality,Barnum-interference}.

{If we take $ \xi $ to be a normalised state $ \rho $, its eigenvalues form a probability distribution, and} we have
$\left\Vert \rho\right\Vert _{\dagger}\leq1$, with equality
if and only if $\rho$ is pure. Note that $\left\Vert \rho\right\Vert _{\dagger}$
is a Schur-convex function \cite{Olkin} of the eigenvalues of $\rho$,
so it is a purity monotone \cite{TowardsThermo}. As such, it attains
its minimum on the invariant state, which is $\left\Vert \chi\right\Vert _{\dagger}=\frac{1}{\sqrt{d}}$,
so for every normalised state one has
\[
\frac{1}{\sqrt{d}}\leq\left\Vert \rho\right\Vert _{\dagger}\leq1,
\]
{consistently with the bounds~\eqref{eq:bounds norms}.}
The square of the dagger norm, still a Schur-convex function, {was} called \emph{purity} in Refs.~\cite{Muller-blackhole,Dahlsten}. Consequently $1-\left\Vert \rho\right\Vert _{\dagger}^{2}$
is a measure of mixedness, sometimes called the \emph{impurity} $I\left(\rho\right)$
of $\rho$. The impurity can be extended to subnormalised states by
defining it as $I\left(\rho\right):=\left(\mathrm{Tr}\:\rho\right)^{2}-\left\Vert \rho\right\Vert _{\dagger}^{2}$
\cite{Barnum-interference}.

The two norms behave differently under channels applied to states.
In Ref.~\cite{Chiribella-purification} it was shown that in causal
theories the operational norm of a state $\rho$ is preserved by channels:
$\left\Vert \mathcal{C}\rho\right\Vert =\left\Vert \rho\right\Vert $
for every channel $\mathcal{C}$, because channels are such that $u\mathcal{C}=u$.

{Instead} the dagger norm shows a different behaviour. To
describe it, it is useful to divide channels into two classes: unital
and non-unital channels \cite{Purity}.
\begin{Definition}
A channel $\mathcal{D}\in\mathsf{Transf}\left(\mathrm{A},\mathrm{B}\right)$
is \emph{unital} if $\mathcal{D}\chi_{\mathrm{A}}=\chi_{\mathrm{B}}$.
\end{Definition}
Unital channels do not increase the dagger norm of states.
\begin{Proposition}
If $\mathcal{D}$ is a unital channel, then $\left\Vert \mathcal{D}\rho\right\Vert _{\dagger}\leq\left\Vert \rho\right\Vert _{\dagger}$,
for every normalised state $\rho$.
\end{Proposition}
\begin{proof}
Unital channels can be chosen as free operations for the resource
theory of purity \cite{Purity}. In Ref.~\cite{Purity} it was shown
that the spectrum of $\mathcal{D}\rho$ is majorised
by the spectrum of $\rho$ (see Ref.~\cite{Olkin} for a definition of majorisation and Schur-convex functions). Since the dagger norm is a Schur-convex
function, we have $\left\Vert \mathcal{D}\rho\right\Vert _{\dagger}\leq\left\Vert \rho\right\Vert _{\dagger}$.
\end{proof}
Clearly if $\mathcal{D}$ is reversible, the dagger norm is preserved,
by Proposition~\ref{prop:invariance reversible}.

For non-unital channels there is at least one state\textemdash the
invariant state $\chi$\textemdash for which the dagger norm increases.
Indeed, if $\mathcal{C}$ is non-unital, $\chi$ is majorised by $\mathcal{C}\chi$,
whence $\left\Vert \chi\right\Vert _{\dagger}\leq\left\Vert \mathcal{C}\chi\right\Vert _{\dagger}$.
Is it true, then, that non-unital channels increase the dagger norm
of all states? The answer is clearly negative. Consider the non-unital
channel mapping all states to a fixed \emph{mixed} state $\rho_{0}\neq\chi$.
For some states, e.g.,\ the invariant state, the dagger norm will
increase, for others, e.g.,\ pure states, the dagger norm will decrease
because it is a purity monotone. In short, for non-unital channels
there is no uniform behaviour of the dagger norm.

\subsection{Dagger fidelity}\label{subsec:dagger fidelity}

The inner product defined in Section~\ref{subsec:Self-duality}
allows us to define a fidelity-like quantity, called the \emph{dagger
fidelity}.
\begin{Definition}
Given two normalised states $\rho$ and $\sigma$, the \emph{dagger
fidelity} is defined as
\[
F_{\dagger}\left(\rho,\sigma\right)=\frac{\left\langle \rho,\sigma\right\rangle }{\left\Vert \rho\right\Vert _{\dagger}\left\Vert \sigma\right\Vert _{\dagger}}.
\]
\end{Definition}
The dagger fidelity measures the overlap between two states. It shares
some properties with the fidelity in quantum theory (cf.\ for instance
Ref.~\cite{Wilde}){, despite \emph{not} coinciding with it.} The first, obvious one, is that $F_{\dagger}\left(\rho,\sigma\right)=F_{\dagger}\left(\sigma,\rho\right)$.

To prove the other properties we need the following lemma, generalising
one of the results of Ref.~\cite{TowardsThermo}.
\begin{Lemma}
\label{lem:perfectly distinguishable}Let $\left\{ \rho_{i}\right\} _{i=1}^{n}$
be perfectly distinguishable states. Then $\left(\rho_{i}^{\dagger}\middle|\rho_{j}\right)=\left\Vert \rho_{i}\right\Vert _{\dagger}^{2}\delta_{ij}$.
\end{Lemma}
\begin{proof}
Clearly what we need to prove is that $\left(\rho_{i}^{\dagger}\middle|\rho_{j}\right)=0$
if $i\neq j$. Let $\left\{ a_{i}\right\} _{i=1}^{n}$ be the perfectly
distinguishing test, and let $\rho_{i}$ be diagonalised as $\rho_{i}=\sum_{k=1}^{{r_i}}{p_{k,i}}\alpha_{k,i}$,
where ${p_{k,i}}>0$ for all $k=1,\ldots,r_i$. We have $\left(a_{i}\middle|\rho_{i}\right)=1$,
hence by Proposition~\ref{prop:non-disturbing} there exists a non-disturbing
pure transformation $\mathcal{T}_{i}$ such that $\mathcal{T}_{i}=_{\rho_{i}}\mathcal{I}$.
Specifically, we have that $\mathcal{T}_{i}\alpha_{k,i}=\alpha_{k,i}$.
Moreover if $i\neq j$, we have $\left(u\middle|\mathcal{T}_{i}\middle|\rho_{j}\right)\leq\left(a_{i}\middle|\rho_{j}\right)=0$,
whence $\left(u\middle|\mathcal{T}_{i}\middle|\rho_{j}\right)=0$. This means that
$\mathcal{T}_{i}\rho_{j}=0$ for all $j\neq i$.

Now, consider
\[
\left(\alpha_{k,i}^{\dagger}\middle|\mathcal{T}_{i}\middle|\alpha_{k,i}\right)=\left(\alpha_{k,i}^{\dagger}\middle|\alpha_{k,i}\right)=1,
\]
where we have used the fact that $\mathcal{T}_{i}\alpha_{k,i}=\alpha_{k,i}$.
Since $\alpha_{k,i}^{\dagger}\mathcal{T}_{i}$ is a pure effect, it
must be $\alpha_{k,i}^{\dagger}\mathcal{T}_{i}=\alpha_{k,i}^{\dagger}$
by Theorem~\ref{thm:duality}. By linearity we have $\rho_{i}^{\dagger}\mathcal{T}_{i}=\rho_{i}^{\dagger}$.
Now, using this fact, for all $j\neq i$
\[
\left(\rho_{i}^{\dagger}\middle|\rho_{j}\right)=\left(\rho_{i}^{\dagger}\middle|\mathcal{T}_{i}\middle|\rho_{j}\right)=0,
\]
because $\mathcal{T}_{i}\rho_{j}=0$.
\end{proof}
Recalling that $\left(\rho^{\dagger}\middle|\sigma\right)=\left\langle \rho,\sigma\right\rangle $,
this lemma means that perfectly distinguishable states form an orthogonal
set. Specifically, if the states are pure, the set is orthonormal.

The following proposition extends and generalises the properties of
the self-dualising inner product of Ref.~\cite{Muller-self-duality}.
\begin{Proposition}
The dagger fidelity has the following properties, for all normalised
states {$ \rho $ and $ \sigma $.}
\begin{enumerate}
\item $0\leq F_{\dagger}\left(\rho,\sigma\right)\leq1$;
\item $F_{\dagger}\left(\rho,\sigma\right)=0$ if and only if $\rho$ and
$\sigma$ are perfectly distinguishable;
\item $F_{\dagger}\left(\rho,\sigma\right)=1$ if and only if $\rho=\sigma$;
\item $F_{\dagger}\left(\mathcal{U}\rho,\mathcal{U}\sigma\right)=F_{\dagger}\left(\rho,\sigma\right)$,
for every reversible channel $\mathcal{U}$.
\end{enumerate}
\end{Proposition}
\begin{proof}
Let us prove the various properties.
\begin{enumerate}
\item Recall that $\left\langle \rho,\sigma\right\rangle =\left(\rho^{\dagger}\middle|\sigma\right)\geq0$,
whence $F_{\dagger}\left(\rho,\sigma\right)\geq0$. Moreover, by Schwarz
inequality, $\left\langle \rho,\sigma\right\rangle \leq\left\Vert \rho\right\Vert _{\dagger}\left\Vert \sigma\right\Vert _{\dagger}$,
so $F_{\dagger}\left(\rho,\sigma\right)\leq1$.
\item Suppose $\rho$ and $\sigma$ are perfectly distinguishable, then
by Lemma~\ref{lem:perfectly distinguishable} $\left\langle \rho,\sigma\right\rangle =0$,
implying $F_{\dagger}\left(\rho,\sigma\right)=0$.
Now suppose $F_{\dagger}\left(\rho,\sigma\right)=0$; then $\left\langle \rho,\sigma\right\rangle =0$.
Let $\rho=\sum_{i=1}^{r}p_{i}\alpha_{i}$ be a diagonalisation of
$\rho$, with $p_{i}>0$, for all $i=1,\ldots,r$, and $r\leq d$.
We have $\sum_{i=1}^{r}p_{i}\left(\alpha_{i}^{\dagger}\middle|\sigma\right)=0$,
which means that $\left(\alpha_{i}^{\dagger}\middle|\sigma\right)=0$ for
$i=1,\ldots,r$. This means that we can build an observation-test
that distinguishes $\rho$ and $\sigma$ perfectly by taking $\left\{ a,u-a\right\} $,
where $a=\sum_{i=1}^{r}\alpha_{i}^{\dagger}$.
\item Clearly, if $\rho=\sigma$, $\left\langle \rho,\sigma\right\rangle =\left\Vert \rho\right\Vert _{\dagger}^{2}$,
whence $F_{\dagger}\left(\rho,\sigma\right)=1$. Conversely, suppose
$F_{\dagger}\left(\rho,\sigma\right)=1$. This means that $\left\langle \rho,\sigma\right\rangle =\left\Vert \rho\right\Vert _{\dagger}\left\Vert \sigma\right\Vert _{\dagger}$.
By Schwarz inequality, this is true if and only if $\rho=\lambda\sigma$,
for some $\lambda\in\mathbb{R}$. Since both states are normalised,
$\lambda=1$, yielding $\rho=\sigma$.
\item This property follows by Proposition~\ref{prop:invariance reversible},
because the inner product and the dagger norm are invariant under
reversible channels.
\end{enumerate}
\end{proof}
Note that Property 3 captures the sharpness of the dagger for all
normalised states \cite{Selby-dagger}.

A property involving tensor product of states is the following.
\begin{Proposition}
\label{prop:fidelity product}For all normalised states $\rho_{1}$,
$\rho_{2}$, $\sigma_{1}$, $\sigma_{2}$ one has
\[
F_{\dagger}\left(\rho_{1}\otimes\rho_{2},\sigma_{1}\otimes\sigma_{2}\right)=F_{\dagger}\left(\rho_{1},\sigma_{1}\right)F_{\dagger}\left(\rho_{2},\sigma_{2}\right)
\]
\end{Proposition}
The proof needs the following easy lemma.
\begin{Lemma}
\label{lem:parallel composition}Let $\rho,\sigma\in\mathsf{St}_{1}\left(\mathrm{A}\right)$,
then $\left(\rho\otimes\sigma\right)^{\dagger}=\rho^{\dagger}\otimes\sigma^{\dagger}$.
\end{Lemma}
\begin{proof}
Let us prove the result for $\rho$ and $\sigma$ pure, the general
result will follow by linearity. By Purity Preservation, $\rho\otimes\sigma$
and $\rho^{\dagger}\otimes\sigma^{\dagger}$ are pure, and one has
$\left(\rho^{\dagger}\otimes\sigma^{\dagger}\middle|\rho\otimes\sigma\right)=1$. By Theorem~\ref{thm:duality}, \linebreak $\left(\rho\otimes\sigma\right)^{\dagger}=\rho^{\dagger}\otimes\sigma^{\dagger}$.
\end{proof}
Now comes the actual proof.
\begin{proof}[Proof of Proposition~\ref{prop:fidelity product}]
We have
\[
F_{\dagger}\left(\rho_{1}\otimes\rho_{2},\sigma_{1}\otimes\sigma_{2}\right)=\frac{\left\langle \rho_{1}\otimes\rho_{2},\sigma_{1}\otimes\sigma_{2}\right\rangle }{\left\Vert \rho_{1}\otimes\rho_{2}\right\Vert _{\dagger}\left\Vert \sigma_{1}\otimes\sigma_{2}\right\Vert _{\dagger}}.
\]
Now, by Lemma~\ref{lem:parallel composition},
\[
\left\langle \rho_{1}\otimes\rho_{2},\sigma_{1}\otimes\sigma_{2}\right\rangle =\left(\rho_{1}^{\dagger}\otimes\rho_{2}^{\dagger}\middle|\sigma_{1}\otimes\sigma_{2}\right)=\left(\rho_{1}^{\dagger}\middle|\sigma_{1}\right)\left(\rho_{2}^{\dagger}\middle|\sigma_{2}\right)=\left\langle \rho_{1},\sigma_{1}\right\rangle \left\langle \rho_{2},\sigma_{2}\right\rangle .
\]
Furthermore,
\[
\left\Vert \rho_{1}\otimes\rho_{2}\right\Vert _{\dagger}=\sqrt{\left\langle \rho_{1}\otimes\rho_{2},\rho_{1}\otimes\rho_{2}\right\rangle }=\sqrt{\left\langle \rho_{1},\rho_{1}\right\rangle \left\langle \rho_{2},\rho_{2}\right\rangle }=\left\Vert \rho_{1}\right\Vert _{\dagger}\left\Vert \rho_{2}\right\Vert _{\dagger}.
\]
Putting everything together,
\[
F_{\dagger}\left(\rho_{1}\otimes\rho_{2},\sigma_{1}\otimes\sigma_{2}\right)=\frac{\left\langle \rho_{1},\sigma_{1}\right\rangle }{\left\Vert \rho_{1}\right\Vert _{\dagger}\left\Vert \sigma_{1}\right\Vert _{\dagger}}\cdot\frac{\left\langle \rho_{2},\sigma_{2}\right\rangle }{\left\Vert \rho_{2}\right\Vert _{\dagger}\left\Vert \sigma_{2}\right\Vert _{\dagger}}=F_{\dagger}\left(\rho_{1},\sigma_{1}\right)F_{\dagger}\left(\rho_{2},\sigma_{2}\right).
\]
\end{proof}

\section{Dagger of all transformations\label{sec:dagger all}}
\setcounter{equation}{0}
\renewcommand{\theequation}{B\arabic{equation}}

Inspired by the results of Lemma~\ref{lem:inner product}, in sharp
theories with purification, we can extend the dagger to all transformations,
a feature often present in process {theories} \cite{selinger2007dagger,Coecke2016,Selby-dagger,Coeckebook}. 
\begin{Definition}
\label{def:dagger transformation}Given the transformation $\mathcal{A}\in\mathsf{Transf}\left(\mathrm{A},\mathrm{B}\right)$,
its \emph{dagger} (or \emph{adjoint}) is a linear transformation $\mathcal{A}^{\dagger}$
from $\mathrm{B}$ to $\mathrm{A}$ defined as\begin{equation}\label{eq:dagger transformation}
\begin{aligned}\Qcircuit @C=1em @R=.7em @!R { & \multiprepareC{1}{\rho} & \qw \poloFantasmaCn{\rB} & \gate{\cA^{\dagger}} & \qw \poloFantasmaCn{\rA} & \qw \\ & \pureghost{\rho} & \qw \poloFantasmaCn{\rS} & \qw &\qw &\qw }\end{aligned}~=\left(\begin{aligned}\Qcircuit @C=1em @R=.7em @!R {  &  \qw \poloFantasmaCn{\rA} & \gate{\cA} & \qw \poloFantasmaCn{\rB} & \multimeasureD{1}{\rho^{\dagger}} \\   & \qw \poloFantasmaCn{\rS} &\qw &\qw  &\ghost{\rho^{\dagger}}}\end{aligned}\right)^{\dagger}~,
\end{equation}for every system $\mathrm{S}$, {and} every state $\rho\in\mathsf{St}_{1}\left(\mathrm{B}\otimes\mathrm{S}\right)$.
\end{Definition}
This definition specifies the dagger of a transformation completely,
thanks to Equation~\eqref{eq:equality transformations}.  Note that Lemma~\ref{lem:inner product}
allows us to formulate Equation~\eqref{eq:dagger symmetric} in term of
effects and their dagger:
\[
\left(a\middle|b^{\dagger}\right)=\left(b\middle|a^{\dagger}\right)
\]
for all effects $a$, and $b$. In this way, Definition~\ref{def:dagger transformation}
can be recast in equivalent terms by taking $b$ as the term in round
brackets in the RHS of Equation~\eqref{eq:dagger transformation}. This
yields\begin{equation}\label{eq:dagger transformation effect}
\begin{aligned}\Qcircuit @C=1em @R=.7em @!R { & \multiprepareC{1}{\rho} & \qw \poloFantasmaCn{\rB} & \gate{\cA^{\dagger}} & \qw \poloFantasmaCn{\rA} & \multimeasureD{1}{E} \\ & \pureghost{\rho} & \qw \poloFantasmaCn{\rS} &\qw &\qw & \ghost{E} }\end{aligned}~=\!\!\!\!\begin{aligned}\Qcircuit @C=1em @R=.7em @!R { & \multiprepareC{1}{E^\dagger} & \qw \poloFantasmaCn{\rA} & \gate{\cA} & \qw \poloFantasmaCn{\rB} & \multimeasureD{1}{\rho^\dagger} \\ & \pureghost{E^\dagger} & \qw \poloFantasmaCn{\rS} &\qw &\qw & \ghost{\rho^\dagger} }\end{aligned}~,
\end{equation}for every system $\mathrm{S}$, every state $\rho\in\mathsf{St}_{1}\left(\mathrm{B}\otimes\mathrm{S}\right)$,
and every effect $E\in\mathsf{Eff}\left(\mathrm{A}\otimes\mathrm{S}\right)$.

The dagger of a transformation may not be a physical transformation,
i.e.,\ it may send physical states to non-physical ones. Indeed, the
action of $\mathcal{A}^{\dagger}\otimes\mathcal{I}$ on a generic
state (the LHS of Equation~\eqref{eq:dagger transformation}) is defined
as the dagger of an effect. However, not all daggers of effects are
physical states. For instance, take the deterministic effect $u=\sum_{i=1}^{d}\alpha_{i}^{\dagger}$,
where $\left\{ \alpha_{i}\right\} _{i=1}^{d}$ is a pure maximal set.
Its dagger is $u^{\dagger}=\sum_{i=1}^{d}\alpha_{i}=d\chi$, which
is a supernormalised (and hence non-physical) state.

For channels, we can give a necessary condition for the existence
of a physical dagger of the channel.
\begin{Proposition}
Let $\mathcal{C}\in\mathsf{Transf}\left(\mathrm{A},\mathrm{B}\right)$
be a channel. If $\mathcal{C}^{\dagger}$ is a physical transformation,
then $\mathcal{C}$ is unital, and $\mathcal{C}^{\dagger}$ itself
is a {unital} channel.
\end{Proposition}
\begin{proof}
If $\mathcal{C}^{\dagger}$ is a physical transformation, then, for
every normalised state $\rho\in\mathsf{St}_{1}\left(\mathrm{B}\right)$,
we have $\left\Vert \mathcal{C}^{\dagger}\rho\right\Vert \leq1$,
or in other words, $\left(u\middle|\mathcal{C}^{\dagger}\middle|\rho\right)\leq1$.
By Equation~\eqref{eq:dagger transformation effect}, $\left(u\middle|\mathcal{C}^{\dagger}\middle|\rho\right)=\left(\rho^{\dagger}\middle|\mathcal{C}\middle|u^{\dagger}\right)$,
so the condition $\left\Vert \mathcal{C}^{\dagger}\rho\right\Vert \leq1$
is equivalent to

\begin{equation}
\left(\rho^{\dagger}\middle|\mathcal{C}\middle|\chi\right)\leq\frac{1}{d},\label{eq:dagger inequality}
\end{equation}
with equality if and only if $\mathcal{C}^{\dagger}$ is a channel.
Suppose by contradiction that $\mathcal{C}$ is not unital, then $\mathcal{C}\chi=\rho_{0}\neq\chi$.
Diagonalise $\rho_{0}$ as $\rho_{0}=\sum_{i=1}^{d}p_{i}\alpha_{i}$,
where $p_{1}\geq p_{2}\geq\ldots\geq p_{d}\geq0$, and $p_{1}>\frac{1}{d}$.
Then taking $\rho$ to be $\alpha_{1}$ in $\left(\rho^{\dagger}\middle|\mathcal{C}\middle|\chi\right)$
yields $p_{1}$, but $p_{1}>\frac{1}{d}$, contradicting Equation~\eqref{eq:dagger inequality}.

Being $\mathcal{C}$ unital, we have that
\[
\left(\rho^{\dagger}\middle|\mathcal{C}\middle|\chi\right)=\left(\rho^{\dagger}\middle|\chi\right)=\frac{1}{d}\mathrm{Tr}\:\rho=\frac{1}{d},
\]
showing that $\mathcal{C}^{\dagger}$ is itself a channel. {Let us prove it is unital. The action of $\mathcal{C}^{\dagger}$ on $ \chi $ is defined in Equation~\eqref{eq:dagger transformation}, so
\[
\mathcal{C}^{\dagger}\chi=\left( \chi^{\dagger}\mathcal{C}\right)^{\dagger}=\frac{1}{d}\left( u\mathcal{C}\right)^{\dagger}=\frac{1}{d}u^\dagger=\chi,
\]
where we have used the fact that $ \mathcal{C} $ is a channel, so $ u\mathcal{C}=u $. This proves that $\mathcal{C}^{\dagger}$ is unital.}
\end{proof}
We can prove that the dagger of a transformation has some nice properties.
\begin{Proposition}
For every transformation $\mathcal{A}\in\mathsf{Transf}\left(\mathrm{A},\mathrm{B}\right)$,
one has $\left(\mathcal{A}^{\dagger}\right)^{\dagger}=\mathcal{A}$.
\end{Proposition}
\begin{proof}
By Equation~\eqref{eq:dagger transformation effect} given any system
$\mathrm{S}$, any state $\rho\in\mathsf{St}_{1}\left({\mathrm{A}}\otimes\mathrm{S}\right)$,
and any effect $E\in\mathsf{Eff}\left({\mathrm{B}}\otimes\mathrm{S}\right)$,
we have\begin{equation}\label{eq:dagger dagger}
\begin{aligned}\Qcircuit @C=1em @R=.7em @!R { & \multiprepareC{1}{\rho} & \qw \poloFantasmaCn{\rA} & \gate{\left(\cA^{\dagger}\right)^{\dagger}} & \qw \poloFantasmaCn{\rB} & \multimeasureD{1}{E} \\ & \pureghost{\rho} & \qw \poloFantasmaCn{\rS} &\qw &\qw & \ghost{E} }\end{aligned}~=\!\!\!\!\begin{aligned}\Qcircuit @C=1em @R=.7em @!R { & \multiprepareC{1}{E^\dagger} & \qw \poloFantasmaCn{\rB} & \gate{\cA^\dagger} & \qw \poloFantasmaCn{\rA} & \multimeasureD{1}{\rho^\dagger} \\ & \pureghost{E^\dagger} & \qw \poloFantasmaCn{\rS} &\qw &\qw & \ghost{\rho^\dagger} }\end{aligned}~.
\end{equation}A linear extension of Equation~\eqref{eq:dagger transformation effect} to cover
the case when $E^{\dagger}$ is not a physical state, applied to the
RHS of Equation~\eqref{eq:dagger dagger} yields\[
\begin{aligned}\Qcircuit @C=1em @R=.7em @!R { & \multiprepareC{1}{E^\dagger} & \qw \poloFantasmaCn{\rB} & \gate{\cA^\dagger} & \qw \poloFantasmaCn{\rA} & \multimeasureD{1}{\rho^\dagger} \\ & \pureghost{E^\dagger} & \qw \poloFantasmaCn{\rS} &\qw &\qw & \ghost{\rho^\dagger} }\end{aligned}~=\!\!\!\!\begin{aligned}\Qcircuit @C=1em @R=.7em @!R { & \multiprepareC{1}{\rho} & \qw \poloFantasmaCn{\rA} & \gate{\cA} & \qw \poloFantasmaCn{\rB} & \multimeasureD{1}{E} \\ & \pureghost{\rho} & \qw \poloFantasmaCn{\rS} &\qw &\qw & \ghost{E} }\end{aligned}.
\]Comparing this with Equation~\eqref{eq:dagger dagger}, we get the thesis.
\end{proof}
We can give a characterisation of the dagger of reversible channels,
which are unital channels.
\begin{Proposition}
If $\mathcal{U}\in\mathsf{Transf}\left(\mathrm{A},\mathrm{B}\right)$
is a reversible channel, $\mathcal{U}^{\dagger}=\mathcal{U}^{-1}$.
\end{Proposition}
\begin{proof}
We have\[
\begin{aligned}\Qcircuit @C=1em @R=.7em @!R { & \multiprepareC{1}{\rho} & \qw \poloFantasmaCn{\rB} & \gate{\cU^{\dagger}} & \qw \poloFantasmaCn{\rA} & \multimeasureD{1}{E} \\ & \pureghost{\rho} & \qw \poloFantasmaCn{\rS} &\qw &\qw & \ghost{E} }\end{aligned}~=\!\!\!\!\begin{aligned}\Qcircuit @C=1em @R=.7em @!R { & \multiprepareC{1}{E^\dagger} & \qw \poloFantasmaCn{\rA} & \gate{\cU} & \qw \poloFantasmaCn{\rB} & \multimeasureD{1}{\rho^\dagger} \\ & \pureghost{E^\dagger} & \qw \poloFantasmaCn{\rS} &\qw &\qw & \ghost{\rho^\dagger} }\end{aligned}~,
\]for any $\mathrm{S}$, $\rho$, $E$. Recalling Lemma~\ref{lem:inner product},
the RHS is $\left\langle \rho,\left(\mathcal{U}\otimes\mathcal{I}\right)E^{\dagger}\right\rangle $.
By Proposition~\ref{prop:invariance reversible} $\left\langle \rho,\left(\mathcal{U}\otimes\mathcal{I}\right)E^{\dagger}\right\rangle =\left\langle \left(\mathcal{U}^{-1}\otimes\mathcal{I}\right)\rho,E^{\dagger}\right\rangle ,$
and by symmetry of the inner product we have that

\[
\left\langle\left(\mathcal{U}^{-1}\otimes\mathcal{I}\right)\rho,E^{\dagger}\right\rangle =\left\langle E^{\dagger},\left(\mathcal{U}^{-1}\otimes\mathcal{I}\right)\rho\right\rangle  =\!\!\!\!\begin{aligned}\Qcircuit @C=1em @R=.7em @!R { & \multiprepareC{1}{\rho} & \qw \poloFantasmaCn{\rB} & \gate{\cU^{-1}} & \qw \poloFantasmaCn{\rA} & \multimeasureD{1}{E} \\ & \pureghost{\rho} & \qw \poloFantasmaCn{\rS} &\qw &\qw & \ghost{E} }\end{aligned}~,
\]whence the thesis follows.
\end{proof}
In particular we have {that} the dagger of the $\mathtt{SWAP}$ channel between
two systems is the $\mathtt{SWAP}$ with the input and output systems
reversed.

The orthogonal projectors of Section~\ref{orthogonal projectors},
on the other hand, are self-adjoint on single system.
\begin{Proposition}
Given the orthogonal projector $P_{\mathsf{I}}$ on a face $F_{\mathsf{I}}$,
we have $P_{\mathsf{I}}^{\dagger}\doteq P_{\mathsf{I}}$.
\end{Proposition}
\begin{proof}
For every $\rho$ and $E$, we have $\left(E\middle|P_{\mathsf{I}}^{\dagger}\middle|\rho\right)=\left(\rho^{\dagger}\middle|P_{\mathsf{I}}\middle|E^{\dagger}\right)$.
The RHS is $\left\langle \rho,P_{\mathsf{I}}E^{\dagger}\right\rangle $.
By the properties of projectors,
\[
\left\langle \rho,P_{\mathsf{I}}E^{\dagger}\right\rangle =\left\langle P_{\mathsf{I}}\rho,E^{\dagger}\right\rangle =\left\langle E^{\dagger},P_{\mathsf{I}}\rho\right\rangle =\left(E\middle|P_{\mathsf{I}}\middle|\rho\right).
\]
This shows that $P_{\mathsf{I}}^{\dagger}\doteq P_{\mathsf{I}}$.
\end{proof}
Finally we prove some properties of the dagger with respect to compositions.
We need an easy lemma first.
\begin{Lemma}
\label{lem:dagger Axi}For every $\mathcal{A}\in\mathsf{Transf}\left(\mathrm{A},\mathrm{B}\right)$,
every system $\mathrm{S},$ and every vector $\xi\in\mathsf{St}_{\mathbb{R}}\left(\mathrm{A}{\otimes}\mathrm{S}\right)$
we have\[
\left(\!\!\!\!\begin{aligned}\Qcircuit @C=1em @R=.7em @!R { & \multiprepareC{1}{\xi} & \qw \poloFantasmaCn{\rA} & \gate{\cA} & \qw \poloFantasmaCn{\rB} & \qw \\ & \pureghost{\xi} & \qw \poloFantasmaCn{\rS} & \qw &\qw &\qw }\end{aligned}~\right)^\dagger=~\begin{aligned}\Qcircuit @C=1em @R=.7em @!R {  &  \qw \poloFantasmaCn{\rB} & \gate{\cA^\dagger} & \qw \poloFantasmaCn{\rA} & \multimeasureD{1}{\xi^{\dagger}} \\   & \qw \poloFantasmaCn{\rS} &\qw &\qw  &\ghost{\xi^{\dagger}}}\end{aligned}~.
\]
\end{Lemma}
\begin{proof}
Recall that $\mathcal{A}=\left(\mathcal{A}^{\dagger}\right)^{\dagger}$;
by Definition~\ref{def:dagger transformation} we have\inputencoding{latin1}{
$\left(\mathcal{A}^{\dagger}\right)^{\dagger}\xi=\left(\xi^{\dagger}\mathcal{A}^{\dagger}\right)^{\dagger}$}\inputencoding{latin9}\[
\begin{aligned}\Qcircuit @C=1em @R=.7em @!R { & \multiprepareC{1}{\xi} & \qw \poloFantasmaCn{\rA} & \gate{\cA} & \qw \poloFantasmaCn{\rB} & \qw \\ & \pureghost{\xi} & \qw \poloFantasmaCn{\rS} & \qw &\qw &\qw }\end{aligned}~=\!\!\!\!\begin{aligned}\Qcircuit @C=1em @R=.7em @!R { & \multiprepareC{1}{\xi} & \qw \poloFantasmaCn{\rA} & \gate{\left(\cA^{\dagger}\right)^\dagger} & \qw \poloFantasmaCn{\rB} & \qw \\ & \pureghost{\xi} & \qw \poloFantasmaCn{\rS} & \qw &\qw &\qw }\end{aligned}~=\left(\begin{aligned}\Qcircuit @C=1em @R=.7em @!R {  &  \qw \poloFantasmaCn{\rB} & \gate{\cA^\dagger} & \qw \poloFantasmaCn{\rA} & \multimeasureD{1}{\xi^{\dagger}} \\   & \qw \poloFantasmaCn{\rS} &\qw &\qw  &\ghost{\xi^{\dagger}}}\end{aligned}\right)^{\dagger}~.
\] Taking the dagger of this equation yields the desired result.
\end{proof}
Now we can state the main results. The first concerns sequential composition.
\begin{Proposition}
\label{prop:dagger sequential}For all transformations $\mathcal{A}\in\mathsf{Transf}\left(\mathrm{A},\mathrm{B}\right)$,
$\mathcal{B}\in\mathsf{Transf}\left(\mathrm{B},\mathrm{C}\right)$, one has $\left(\mathcal{B}\mathcal{A}\right)^{\dagger}=\mathcal{A}^{\dagger}\mathcal{B}^{\dagger}$.
\end{Proposition}
\begin{proof}
Take any system $\mathrm{S}$, any state $\rho\in\mathsf{St}_{1}\left(\mathrm{C}\otimes\mathrm{S}\right)$,
and any effect $E\in\mathsf{Eff}\left(\mathrm{A}\otimes\mathrm{S}\right)$.
By Equation~\eqref{eq:dagger transformation effect} we have\[
\begin{aligned}\Qcircuit @C=1em @R=.7em @!R { & \multiprepareC{1}{\rho} & \qw \poloFantasmaCn{\rC} & \gate{\left(\cB\cA\right)^{\dagger}} & \qw \poloFantasmaCn{\rA} & \multimeasureD{1}{E} \\ & \pureghost{\rho} & \qw \poloFantasmaCn{\rS} &\qw &\qw & \ghost{E} }\end{aligned}~=\!\!\!\!\begin{aligned}\Qcircuit @C=1em @R=.7em @!R { & \multiprepareC{1}{E^\dagger} & \qw \poloFantasmaCn{\rA} & \gate{\cB\cA} & \qw \poloFantasmaCn{\rC} & \multimeasureD{1}{\rho^\dagger} \\ & \pureghost{E^\dagger} & \qw \poloFantasmaCn{\rS} &\qw &\qw & \ghost{\rho^\dagger} }\end{aligned}~=\!\!\!\!\begin{aligned}\Qcircuit @C=1em @R=.7em @!R { & \multiprepareC{1}{E^\dagger} & \qw \poloFantasmaCn{\rA} & \gate{\cA} & \qw \poloFantasmaCn{\rB} &\gate{\cB}&\qw \poloFantasmaCn{\rC}& \multimeasureD{1}{\rho^\dagger} \\ & \pureghost{E^\dagger} & \qw \poloFantasmaCn{\rS} &\qw &\qw &\qw &\qw & \ghost{\rho^\dagger} }\end{aligned}~.
\]Define $\xi$ as $\xi:=\left(\mathcal{A}\otimes\mathcal{I}\right)E^{\dagger}$,
so\[
\begin{aligned}\Qcircuit @C=1em @R=.7em @!R { & \multiprepareC{1}{\rho} & \qw \poloFantasmaCn{\rC} & \gate{\left(\cB\cA\right)^{\dagger}} & \qw \poloFantasmaCn{\rA} & \multimeasureD{1}{E} \\ & \pureghost{\rho} & \qw \poloFantasmaCn{\rS} &\qw &\qw & \ghost{E} }\end{aligned}~=\!\!\!\!\begin{aligned}\Qcircuit @C=1em @R=.7em @!R { & \multiprepareC{1}{\xi} & \qw \poloFantasmaCn{\rB} & \gate{\cB} & \qw \poloFantasmaCn{\rC} & \multimeasureD{1}{\rho^\dagger} \\ & \pureghost{\xi} & \qw \poloFantasmaCn{\rS} &\qw &\qw & \ghost{\rho^\dagger} }\end{aligned}~=\!\!\!\!\begin{aligned}\Qcircuit @C=1em @R=.7em @!R { & \multiprepareC{1}{\rho} & \qw \poloFantasmaCn{\rC} & \gate{\cB^\dagger} & \qw \poloFantasmaCn{\rB} & \multimeasureD{1}{\xi^\dagger} \\ & \pureghost{\rho} & \qw \poloFantasmaCn{\rS} &\qw &\qw & \ghost{\xi^\dagger} }\end{aligned}~.
\]By Lemma~\ref{lem:dagger Axi} $\xi^{\dagger}=\left[\left(\mathcal{A}\otimes\mathcal{I}\right)E^{\dagger}\right]^{\dagger}=E\left(\mathcal{A}^{\dagger}\otimes\mathcal{I}\right)$,
then\[
\begin{aligned}\Qcircuit @C=1em @R=.7em @!R { & \multiprepareC{1}{\rho} & \qw \poloFantasmaCn{\rC} & \gate{\left(\cB\cA\right)^{\dagger}} & \qw \poloFantasmaCn{\rA} & \multimeasureD{1}{E} \\ & \pureghost{\rho} & \qw \poloFantasmaCn{\rS} &\qw &\qw & \ghost{E} }\end{aligned}~=\!\!\!\!\begin{aligned}\Qcircuit @C=1em @R=.7em @!R { & \multiprepareC{1}{\rho} & \qw \poloFantasmaCn{\rC} & \gate{\cB^\dagger} & \qw \poloFantasmaCn{\rB} &\gate{\cA^\dagger}&\qw \poloFantasmaCn{\rA}& \multimeasureD{1}{E} \\ & \pureghost{\rho} & \qw \poloFantasmaCn{\rS} &\qw &\qw &\qw &\qw & \ghost{E} }\end{aligned}~,
\]therefore $\left(\mathcal{B}\mathcal{A}\right)^{\dagger}=\mathcal{A}^{\dagger}\mathcal{B}^{\dagger}$.
\end{proof}
Finally the dagger respects parallel composition. Again we need a
lemma.
\begin{Lemma}
\label{lem:dagger IAI}For every $\mathcal{A}\in\mathsf{Transf}\left(\mathrm{A},\mathrm{B}\right)$,
every systems $\mathrm{S}$ and $\mathrm{S}'$, we have $\left(\mathcal{I}_{\mathrm{S}}\otimes\mathcal{A}\otimes\mathcal{I}_{\mathrm{S}'}\right)^{\dagger}=\mathcal{I}_{\mathrm{S}}\otimes\mathcal{A}^{\dagger}\otimes\mathcal{I}_{\mathrm{S}'}$.
\end{Lemma}
\begin{proof}
As a first step, let us prove that, for every system $\mathrm{S}$,
we have $\left(\mathcal{A}\otimes\mathcal{I}_{\mathrm{S}}\right)^{\dagger}=\mathcal{A}^{\dagger}\otimes\mathcal{I}_{\mathrm{S}}$.
Take any system $\mathrm{S}'$, any state $\rho\in\mathsf{St}_{1}\left(\mathrm{B}\otimes\mathrm{S}\otimes\mathrm{S'}\right)$,
and any effect $E\in\mathsf{Eff}\left(\mathrm{A}\otimes\mathrm{S}\otimes\mathrm{S'}\right)$,
Equation~\eqref{eq:dagger transformation effect} yields

\[
\begin{aligned}\Qcircuit @C=1em @R=.7em @!R { & \multiprepareC{2}{\rho} & \qw \poloFantasmaCn{\rB} & \multigate{1}{\left(\cA\otimes\cI\right)^{\dagger}} & \qw \poloFantasmaCn{\rA} & \multimeasureD{2}{E} \\ & \pureghost{\rho} & \qw \poloFantasmaCn{\rS} &\ghost{\left(\cA\otimes\cI\right)^{\dagger}} &\qw \poloFantasmaCn{\rS} & \ghost{E} \\ & \pureghost{\rho} & \qw \poloFantasmaCn{\rS'} &\qw &\qw & \ghost{E}}\end{aligned}~=\!\!\!\!\begin{aligned}\Qcircuit @C=1em @R=.7em @!R { & \multiprepareC{2}{E^\dagger} & \qw \poloFantasmaCn{\rA} & \multigate{1}{\cA\otimes\cI} & \qw \poloFantasmaCn{\rB} & \multimeasureD{2}{\rho^\dagger} \\ & \pureghost{E^\dagger} & \qw \poloFantasmaCn{\rS} &\ghost{\cA\otimes\cI} &\qw \poloFantasmaCn{\rS} & \ghost{\rho^\dagger} \\ & \pureghost{E^\dagger} & \qw \poloFantasmaCn{\rS'} &\qw &\qw & \ghost{\rho^\dagger}}\end{aligned}~=\!\!\!\!\begin{aligned}\Qcircuit @C=1em @R=.7em @!R { & \multiprepareC{2}{E^\dagger} & \qw \poloFantasmaCn{\rA} & \gate{\cA} & \qw \poloFantasmaCn{\rB} & \multimeasureD{2}{\rho^\dagger} \\ & \pureghost{E^\dagger} & \qw \poloFantasmaCn{\rS} &\qw &\qw  & \ghost{\rho^\dagger} \\ & \pureghost{E^\dagger} & \qw \poloFantasmaCn{\rS'} &\qw &\qw & \ghost{\rho^\dagger}}\end{aligned}~.
\]Specialising Equation~\eqref{eq:dagger transformation effect} to the
case {of a composite system}, we have\[
\begin{aligned}\Qcircuit @C=1em @R=.7em @!R { & \multiprepareC{2}{E^\dagger} & \qw \poloFantasmaCn{\rA} & \gate{\cA} & \qw \poloFantasmaCn{\rB} & \multimeasureD{2}{\rho^\dagger} \\ & \pureghost{E^\dagger} & \qw \poloFantasmaCn{\rS} &\qw &\qw  & \ghost{\rho^\dagger} \\ & \pureghost{E^\dagger} & \qw \poloFantasmaCn{\rS'} &\qw &\qw & \ghost{\rho^\dagger}}\end{aligned}~=\!\!\!\!\begin{aligned}\Qcircuit @C=1em @R=.7em @!R { & \multiprepareC{2}{\rho} & \qw \poloFantasmaCn{\rB} & \gate{\cA^\dagger} & \qw \poloFantasmaCn{\rA} & \multimeasureD{2}{E} \\ & \pureghost{\rho} & \qw \poloFantasmaCn{\rS} &\qw &\qw  & \ghost{E} \\ & \pureghost{\rho} & \qw \poloFantasmaCn{\rS'} &\qw &\qw & \ghost{E}}\end{aligned}~,
\]whence we conclude that $\left(\mathcal{A}\otimes\mathcal{I}_{\mathrm{S}}\right)^{\dagger}=\mathcal{A}^{\dagger}\otimes\mathcal{I}_{\mathrm{S}}$.

Now let us prove that, for every system $\mathrm{S}$, $\left(\mathcal{I}_{\mathrm{S}}\otimes\mathcal{A}\right)^{\dagger}=\mathcal{I}_{\mathrm{S}}\otimes\mathcal{A}^{\dagger}$.
Note that \[
\begin{aligned}\Qcircuit @C=1em @R=.7em @!R {&\qw\poloFantasmaCn{\rS}&\qw&\qw&\qw\\&\qw\poloFantasmaCn{\rA}&\gate{\cA}&\qw\poloFantasmaCn{\rB}&\qw}\end{aligned}~=~\begin{aligned}\Qcircuit @C=1em @R=.7em @!R {&\qw\poloFantasmaCn{\rS}&\multigate{1}{\mathtt{SWAP}}&\qw \poloFantasmaCn{\rA} &\gate{\cA} &\qw \poloFantasmaCn{\rB} &\multigate{1}{\mathtt{SWAP}} &\qw \poloFantasmaCn{\rS} &\qw \\ &\qw\poloFantasmaCn{\rA}&\ghost{\mathtt{SWAP}}&\qw \poloFantasmaCn{\rS} &\qw &\qw&\ghost{\mathtt{SWAP}}&\qw \poloFantasmaCn{\rB} &\qw}\end{aligned}~.
\]By Proposition~\ref{prop:dagger sequential}, {and} recalling what we have
just proved, we have\[
\left(\begin{aligned}\Qcircuit @C=1em @R=.7em @!R {&\qw\poloFantasmaCn{\rS}&\qw&\qw&\qw\\&\qw\poloFantasmaCn{\rA}&\gate{\cA}&\qw\poloFantasmaCn{\rB}&\qw}\end{aligned}\right)^\dagger~=~\begin{aligned}\Qcircuit @C=1em @R=.7em @!R {&\qw\poloFantasmaCn{\rS}&\multigate{1}{\mathtt{SWAP}}&\qw \poloFantasmaCn{\rB} &\gate{\cA^\dagger} &\qw \poloFantasmaCn{\rA} &\multigate{1}{\mathtt{SWAP}} &\qw \poloFantasmaCn{\rS} &\qw \\ &\qw\poloFantasmaCn{\rB}&\ghost{\mathtt{SWAP}}&\qw \poloFantasmaCn{\rS} &\qw &\qw&\ghost{\mathtt{SWAP}}&\qw \poloFantasmaCn{\rA} &\qw}\end{aligned}~=~\begin{aligned}\Qcircuit @C=1em @R=.7em @!R {&\qw\poloFantasmaCn{\rS}&\qw&\qw&\qw\\&\qw\poloFantasmaCn{\rB}&\gate{\cA^\dagger}&\qw\poloFantasmaCn{\rA}&\qw}\end{aligned}~.
\]

To get the thesis, note that $\left(\mathcal{I}_{\mathrm{S}}\otimes\mathcal{A}\otimes\mathcal{I}_{\mathrm{S}'}\right)^{\dagger}=\left[\left(\mathcal{I}_{\mathrm{S}}\otimes\mathcal{A}\right)\otimes\mathcal{I}_{\mathrm{S}'}\right]^{\dagger}$.
We have just proved that
\[
\left[\left(\mathcal{I}_{\mathrm{S}}\otimes\mathcal{A}\right)\otimes\mathcal{I}_{\mathrm{S}'}\right]^{\dagger}=\left(\mathcal{I}_{\mathrm{S}}\otimes\mathcal{A}\right)^{\dagger}\otimes\mathcal{I}_{\mathrm{S}'},
\]
and that $\left(\mathcal{I}_{\mathrm{S}}\otimes\mathcal{A}\right)^{\dagger}=\mathcal{I}_{\mathrm{S}}\otimes\mathcal{A}^{\dagger}$,
therefore we conclude that $\left(\mathcal{I}_{\mathrm{S}}\otimes\mathcal{A}\otimes\mathcal{I}_{\mathrm{S}'}\right)^{\dagger}=\mathcal{I}_{\mathrm{S}}\otimes\mathcal{A}^{\dagger}\otimes\mathcal{I}_{\mathrm{S}'}$.
\end{proof}
\begin{Proposition}
Let $\mathcal{A}\in\mathsf{Transf}\left(\mathrm{A},\mathrm{B}\right)$,
and $\mathcal{B}\in\mathsf{Transf}\left(\mathrm{C},\mathrm{D}\right)$.
We have $\left(\mathcal{A}\otimes\mathcal{B}\right)^{\dagger}=\mathcal{A}^{\dagger}\otimes\mathcal{B}^{\dagger}$.
\end{Proposition}
\begin{proof}
Take any system $\mathrm{S}$, any state $\rho\in\mathsf{St}_{1}\left(\mathrm{B}\otimes\mathrm{D}\otimes\mathrm{S}\right)$,
and any effect $E\in\mathsf{Eff}\left(\mathrm{A}\otimes\mathrm{C}\otimes\mathrm{S}\right)$,
we have

\[
\begin{aligned}\Qcircuit @C=1em @R=.7em @!R { & \multiprepareC{2}{\rho} & \qw \poloFantasmaCn{\rB} & \multigate{1}{\left(\cA\otimes\cB\right)^{\dagger}} & \qw \poloFantasmaCn{\rA} & \multimeasureD{2}{E} \\ & \pureghost{\rho} & \qw \poloFantasmaCn{\rD} &\ghost{\left(\cA\otimes\cB\right)^{\dagger}} &\qw \poloFantasmaCn{\rC} & \ghost{E} \\ & \pureghost{\rho} & \qw \poloFantasmaCn{\rS} &\qw &\qw & \ghost{E}}\end{aligned}~=\!\!\!\!\begin{aligned}\Qcircuit @C=1em @R=.7em @!R { & \multiprepareC{2}{E^\dagger} & \qw \poloFantasmaCn{\rA} & \multigate{1}{\cA\otimes\cB} & \qw \poloFantasmaCn{\rB} & \multimeasureD{2}{\rho^\dagger} \\ & \pureghost{E^\dagger} & \qw \poloFantasmaCn{\rC} &\ghost{\cA\otimes\cB} &\qw \poloFantasmaCn{\rD} & \ghost{\rho^\dagger} \\ & \pureghost{E^\dagger} & \qw \poloFantasmaCn{\rS} &\qw &\qw & \ghost{\rho^\dagger}}\end{aligned}~=\!\!\!\!\begin{aligned}\Qcircuit @C=1em @R=.7em @!R { & \multiprepareC{2}{E^\dagger} & \qw \poloFantasmaCn{\rA} & \gate{\cA} & \qw \poloFantasmaCn{\rB} & \multimeasureD{2}{\rho^\dagger} \\ & \pureghost{E^\dagger} & \qw \poloFantasmaCn{\rC} &\gate{\cB} &\qw \poloFantasmaCn{\rD} & \ghost{\rho^\dagger} \\ & \pureghost{E^\dagger} & \qw \poloFantasmaCn{\rS} &\qw &\qw & \ghost{\rho^\dagger}}\end{aligned}~.
\]Now define $\xi:=\left(\mathcal{I}_{\mathrm{A}}\otimes\mathcal{B}\otimes\mathcal{I}_{\mathrm{S}}\right)E^{\dagger}$,
hence\[
\begin{aligned}\Qcircuit @C=1em @R=.7em @!R { & \multiprepareC{2}{\rho} & \qw \poloFantasmaCn{\rB} & \multigate{1}{\left(\cA\otimes\cB\right)^{\dagger}} & \qw \poloFantasmaCn{\rA} & \multimeasureD{2}{E} \\ & \pureghost{\rho} & \qw \poloFantasmaCn{\rD} &\ghost{\left(\cA\otimes\cB\right)^{\dagger}} &\qw \poloFantasmaCn{\rC} & \ghost{E} \\ & \pureghost{\rho} & \qw \poloFantasmaCn{\rS} &\qw &\qw & \ghost{E}}\end{aligned}~=\!\!\!\!\begin{aligned}\Qcircuit @C=1em @R=.7em @!R { & \multiprepareC{2}{\xi} & \qw \poloFantasmaCn{\rA} & \gate{\cA} & \qw \poloFantasmaCn{\rB} & \multimeasureD{2}{\rho^\dagger} \\ & \pureghost{\xi} & \qw \poloFantasmaCn{\rD} &\qw &\qw  & \ghost{\rho^\dagger} \\ & \pureghost{\xi} & \qw \poloFantasmaCn{\rS} &\qw &\qw & \ghost{\rho^\dagger}}\end{aligned}~=\!\!\!\!\begin{aligned}\Qcircuit @C=1em @R=.7em @!R { & \multiprepareC{2}{\rho} & \qw \poloFantasmaCn{\rB} & \gate{\cA^\dagger} & \qw \poloFantasmaCn{\rA} & \multimeasureD{2}{\xi^\dagger} \\ & \pureghost{\rho} & \qw \poloFantasmaCn{\rD} &\qw &\qw  & \ghost{\xi^\dagger} \\ & \pureghost{\rho} & \qw \poloFantasmaCn{\rS} &\qw &\qw & \ghost{\xi^\dagger}}\end{aligned}
\]By Lemmas~\ref{lem:dagger Axi} and \ref{lem:dagger IAI}, we have
that $\xi^{\dagger}=E\left(\mathcal{I}_{\mathrm{A}}\otimes\mathcal{B}^{\dagger}\otimes\mathcal{I}_{\mathrm{S}}\right)$,
so\[
\begin{aligned}\Qcircuit @C=1em @R=.7em @!R { & \multiprepareC{2}{\rho} & \qw \poloFantasmaCn{\rB} & \multigate{1}{\left(\cA\otimes\cB\right)^{\dagger}} & \qw \poloFantasmaCn{\rA} & \multimeasureD{2}{E} \\ & \pureghost{\rho} & \qw \poloFantasmaCn{\rD} &\ghost{\left(\cA\otimes\cB\right)^{\dagger}} &\qw \poloFantasmaCn{\rC} & \ghost{E} \\ & \pureghost{\rho} & \qw \poloFantasmaCn{\rS} &\qw &\qw & \ghost{E}}\end{aligned}~=\!\!\!\!\begin{aligned}\Qcircuit @C=1em @R=.7em @!R { & \multiprepareC{2}{\rho} & \qw \poloFantasmaCn{\rB} & \gate{\cA^\dagger} & \qw \poloFantasmaCn{\rA} & \multimeasureD{2}{E} \\ & \pureghost{\rho} & \qw \poloFantasmaCn{\rD} &\gate{\cB^\dagger} &\qw \poloFantasmaCn{\rC} & \ghost{E} \\ & \pureghost{\rho} & \qw \poloFantasmaCn{\rS} &\qw &\qw & \ghost{E}}\end{aligned}~,
\]whence the thesis.
\end{proof}
This means that the dagger respects the composition of diagrams, and
corresponds to the action of flipping a diagram with respect to a
vertical axis.

\bibliographystyle{plain}
\bibliography{Bibliography}
\end{document}

%% file: myQcircuit.tex
\usepackage[matrix,frame,arrow]{xy}
\usepackage{amsmath}

\newcommand{\qw}[1][-1]{\ar @{-} [0,#1]}

\newcommand{\gate}[1]{*{\xy *+<.6em>{#1};p\save+LU;+RU **\dir{-}\restore\save+RU;+RD **\dir{-}\restore\save+RD;+LD **\dir{-}\restore\POS+LD;+LU **\dir{-}\endxy} \qw}

\newcommand{\measureD}[1]{*{\xy*+=+<.5em>{\vphantom{\rule{0em}{.1em}#1}}*\cir{r_l};p\save*!R{#1} \restore\save+UC;+UC-<.5em,0em>*!R{\hphantom{#1}}+L **\dir{-} \restore\save+DC;+DC-<.5em,0em>*!R{\hphantom{#1}}+L **\dir{-} \restore\POS+UC-<.5em,0em>*!R{\hphantom{#1}}+L;+DC-<.5em,0em>*!R{\hphantom{#1}}+L **\dir{-} \endxy} \qw}

\newcommand{\multimeasureD}[2]{*+<1em,.9em>{\hphantom{#2}}\save[0,0].[#1,0];p\save !C *{#2},p+LU+<0em,0em>;+RU+<-.8em,0em> **\dir{-}\restore\save +LD;+LU **\dir{-}\restore\save +LD;+RD-<.8em,0em> **\dir{-} \restore\save +RD+<0em,.8em>;+RU-<0em,.8em> **\dir{-} \restore \POS !UR*!UR{\cir<.9em>{r_d}};!DR*!DR{\cir<.9em>{d_l}}\restore \qw}

\newcommand{\multigate}[2]{*+<1em,.9em>{\hphantom{#2}} \qw \POS[0,0].[#1,0];p !C *{#2},p \save+LU;+RU **\dir{-}\restore\save+RU;+RD **\dir{-}\restore\save+RD;+LD **\dir{-}\restore\save+LD;+LU **\dir{-}\restore}
\newcommand{\ghost}[1]{*+<1em,.9em>{\hphantom{#1}} \qw}
\newcommand{\Qcircuit}[1][0em]{\xymatrix @*=<#1>} 

\newcommand{\pureghost}[1]{*+<1em,.9em>{\hphantom{#1}}}
\newcommand{\multiprepareC}[2]{*+<1em,.9em>{\hphantom{#2}}\save[0,0].[#1,0];p\save !C
  *{#2},p+RU+<0em,0em>;+LU+<+.8em,0em> **\dir{-}\restore\save +RD;+RU **\dir{-}\restore\save
  +RD;+LD+<.8em,0em> **\dir{-} \restore\save +LD+<0em,.8em>;+LU-<0em,.8em> **\dir{-} \restore \POS
  !UL*!UL{\cir<.9em>{u_r}};!DL*!DL{\cir<.9em>{l_u}}\restore}
\newcommand{\prepareC}[1]{*{\xy*+=+<.5em>{\vphantom{#1\rule{0em}{.1em}}}*\cir{l^r};p\save*!L{#1} \restore\save+UC;+UC+<.5em,0em>*!L{\hphantom{#1}}+R **\dir{-} \restore\save+DC;+DC+<.5em,0em>*!L{\hphantom{#1}}+R **\dir{-} \restore\POS+UC+<.5em,0em>*!L{\hphantom{#1}}+R;+DC+<.5em,0em>*!L{\hphantom{#1}}+R **\dir{-} \endxy}}
\newcommand{\poloFantasmaCn}[1]{{{}^{#1}_{\phantom{#1}}}}